\documentclass[journal,transmag]{IEEEtran}

\usepackage{mathtools}
\usepackage[mathscr]{euscript}
\newcommand{\beq}{\begin{equation}}
\newcommand{\eeq}{\end{equation}}
\newcommand{\beqq}{\begin{equation*}}
\newcommand{\eeqq}{\end{equation*}}
\newcommand{\ei}{\end{itemize}}
\newcommand{\bi}{\begin{itemize}}
\usepackage{graphics}
\usepackage{placeins}
\usepackage[british]{babel}
\usepackage{graphicx}
\usepackage[dvipsnames]{xcolor}
\usepackage{caption}
\usepackage{amsthm}
\usepackage{subcaption}
\usepackage{color}
\usepackage{tabularx,bbm}
\usepackage{cite}
\usepackage{url}
\usepackage{amsmath,amssymb,amscd}
\usepackage{footnote}

\newcommand{\E}{\mathbbm{E}}

\newtheorem{theorem}{Theorem}[section]

\newtheorem{corollary}{Corollary}
\newtheorem{definition}{Definition}
\newtheorem{lemma}{Lemma}

\newtheorem{remark}{Remark}

\newif\ifnotesw
\noteswtrue

\newif\ifnotesw
\noteswtrue

\def\E{\mathbf{E}}

\def\beq{\begin{equation}}
\def\eeq{\end{equation}}
\def\beqn{\begin{eqnarray}}
\def\eeqn{\end{eqnarray}}

\def\CH{{\mathcal H}}
\def\CV{{\mathcal V}}
\begin{document}
\title{Information Dissemination Speed in Delay Tolerant Urban\\ Vehicular Networks
in a Hyperfractal Setting\\ }

\author{Dalia~Popescu, 
        Philippe~Jacquet,
        Bernard~Mans, 
        Robert~Dumitru,
        Andra~Pastrav,
        and~Emanuel~Puschita 
\thanks{Part of the work has been done at Lincs.}
\thanks{Manuscript received July, 2018; revised \today.}}

%

\markboth{Journal of \LaTeX\ Class Files,~Vol.~xx, No.~x, August~2018}%
{Shell \MakeLowercase{\textit{et al.}}: Bare Demo of IEEEtran.cls for IEEE Journals}
%

\maketitle

\begin{abstract}
This paper studies the fundamental communication properties of urban vehicle networks by exploiting the self-similarity and hierarchical organization of modern cities. We use an innovative model called ``hyperfractal'' that captures the self-similarities of both the traffic and vehicle locations but avoids the extremes of regularity and randomness. 
We use analytical tools to derive theoretical upper and lower bounds for the information
propagation speed in an urban delay tolerant network (i.e., a network that is disconnected at all time, and thus uses a store-carry-and-forward routing model). We prove that the average broadcast time behaves as $n^{1-\delta}$ times a slowly varying function, where $\delta$  depends on the precise fractal dimension.

Furthermore, we show that the broadcast speedup is due in part  to an interesting self-similar phenomenon, that we denote as {\em information teleportation}. This phenomenon arises as a consequence of the topology of the vehicle traffic, and triggers an acceleration of the broadcast time. We show that our model fits real cities where open traffic data sets are available. We present simulations 
confirming  the validity of the bounds in multiple realistic settings, including scenarios with variable speed, using both QualNet and a discrete-event simulator in Matlab.

\end{abstract}



\IEEEpeerreviewmaketitle

\section{Introduction}

We are now on the verge of a new industrial revolution sparked by the Internet of Things. 
With the aim of making cities smarter and people's life easier, the new communication scenarios require enhanced performances and force us to rethink the way we design and analyze networks. These will incorporate the characteristics of human society 
with an ever increasing amount of learning through the emergence of artificial intelligence. 

Smart cities will include a tremendous number of connected devices and heterogeneous communication scenarios: drones communicating to smart buildings, buses to bus stops, ambulances to bicycles, vehicles to traffic lights and between themselves. 
Either in their most simple form, vehicle to vehicle communications, or in an extended form, vehicle to everything (e.g., infrastructure), the vehicular networks are a major part of the new communication ecosystem arising in a smart city.

Distributed networks of vehicles such as vehicular ad-hoc networks (VANET) can easily be turned into an infrastructure-less self-organizing traffic information system, where any vehicle can participate in collecting and reporting information. 
As the number of vehicular networks continue to grow and now create giant networks (with diverse hierarchical structures and node types), vehicular interactions are becoming more complex. This complexity is exacerbated by the time-space relationships between vehicles.
The intrinsic mobility of the vehicles on the roads leads to highly dynamic and evolutionary topologies that can no longer be adequately modeled through methods inherited from previous networks generations. 


As the connected vehicles are following the deployment of human activity, a flexible model should capture the characteristics of the dynamics of human society. Confinement of human settlement in areas limited in size is the foundation of the long-standing Central Place Theory (CPT) which assumes the existence of regular spatial patterns in regional human organization \cite{nature}. For example, a county or department has rural areas with low density of population, while urban areas have high density of population, namely cities and towns. Similarly, cities reflect a statistical self-similarity or hierarchy of clusters \cite{batty1}. The towns are split in neighborhoods, each neighborhood is organized in quarters, then blocks separated by streets. Blocks are made of buildings that are themselves split in apartments and so on.
The concept of self-similarity, very present in nature, has been exploited for a long time already in urban planning and architecture \cite{batty1,batty2}.
As streets are located between buildings, the structure of streets and roads in a city, and thus traffic inherit the self-similar nature of urban architecture. Figure~\ref{fig:data} is an illustration of the spatial patterns appearing in the traffic in Seattle. 

\begin{figure}[httb]\centering
\includegraphics[scale=0.35, trim=0cm 0.5cm 0cm 0cm]{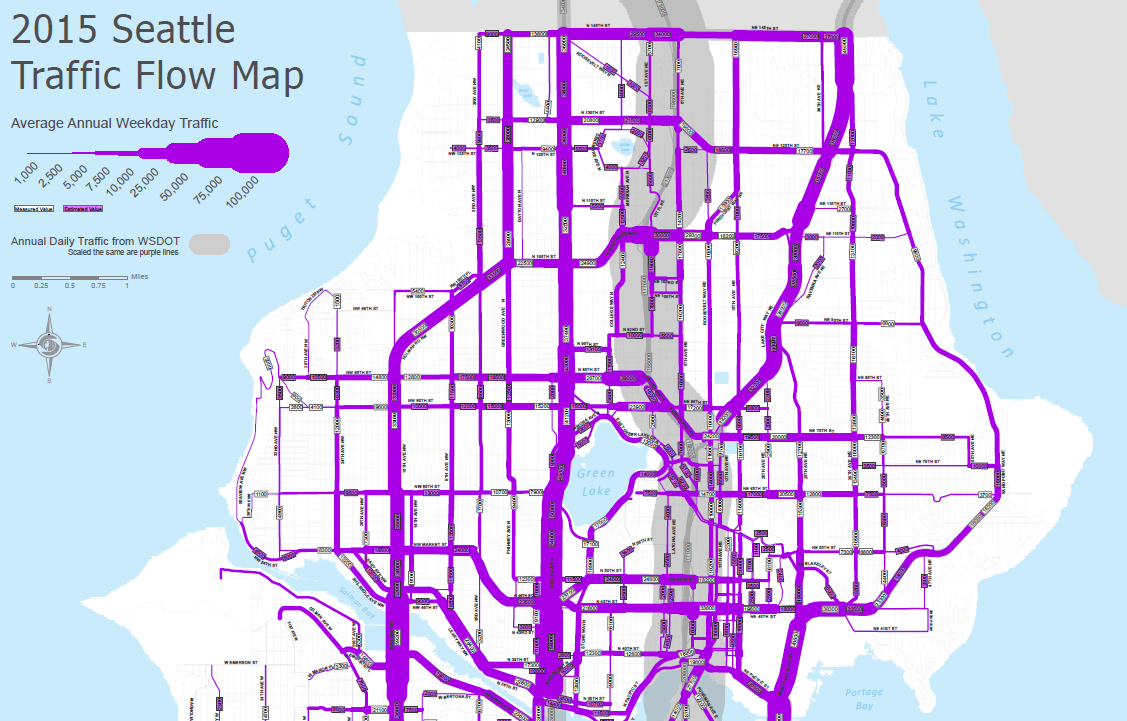}
\caption{Seattle traffic snapshot}
\vspace{-0.09cm}
\label{fig:data}
\end{figure}

The concept of self-similarity is particularly important in mathematics: a self-similar object, or fractal, is an   object in geometry which shows strong similarity (either exactly or approximately) with smaller parts of itself ({\it i.e.} the whole has the same shape as one or more of the parts). Fractals are objects that commonly exhibit similar patterns at increasingly small scales and are commonly used to describe and simulate objects occurring in nature (see the seminal work of Mandelbrot, e.g. \cite{mandelbrot}). 
Fractals have been recently introduced for wireless network topology modeling \cite{fractal} or for cellular coverage modeling \cite{fractal_2}. The models of spatial repartition of population following a fractal (or hyperfractal) distribution have been introduced in order to provide a more realistic and complete description of geometric properties and interactions that arise in an urban ad-hoc wireless network \cite{fractal, spaswin}.

Disseminating information in a network is typi\-cally done in a broadcast-oriented way \cite{new_broadcast1,new_broadcast_safety}. 
Broadcast in vehicular networks can be done with the help of the adjacent infrastructure, yet, in this work,  we focus on the study of infrastructure-less vehicular networks to understand their performance limits. 
Broadcasting schemes, rather than unicasting, are particularly adequate for vehicular networks due to the mobility of the cars, which implies an ongoing evolving topology. In broadcast protocols, the advantage stands in the fact that the vehicles do not require the knowledge of a specific destination location or its relevant route. This eliminates the complexity of route discovery, address re\-solution and the previously mentioned topology management with mobility, which are critical aspects in dynamic networks such as vehicular networks. Broadcast protocols have been enhanced throughout the past years in order to minimize the redundancy and energy consumption, to reduce the security issues \cite{broadcast_security}, increase reliability and comparisons have been made between the achievable performances with and without road-side infrastructure \cite{new_broadcast_safety}.

 In our previous works, \cite{spaswin,gsi} we have introduced the hyperfractal model for the vehicular traffic densities on streets. The hyperfractal has the unique novelty and advantage that it captures the self-similarity of the topology and provides a flexible yet robust model for vehicular networks in urban sce\-nario. 
In this work, in order to extend our understanding of the hyperfractal model and further demonstrate its capabilities, we study the time limit requirements for a piece of information to propagate in an urban vehicular network, by exploiting the model. 
 

On one hand, our aim is to better understand the impact on the broadcast time of the particular environment, the traffic and network topologies. We  show that network geometry should indeed be considered and exploited when designing broadcast protocols. On the other hand, we aim at demonstrating the ease of use of the hyperfractal model and how it can be exploited to enhance the computation of wireless networks Key Performance Indicators (KPIs). For instance, most of broadcast protocols parameters are dependent on traffic density and this is a metric which can be predicted when considering the particular network geometry. 

We prove that the
average broadcast time in a hyperfractal setup is in $O(n^{1-\delta} \log n)$
where $n$ is the number of mobile nodes and where $\delta$  depends on the precise fractal dimension. Furthermore, we provide an intuitive procedure that allows the transformation of traffic flow maps of neighborhood and cities into hyperfractals, facilitating future use of the model by the community.

\section{Related Work}
\label{relatedWork}
The research community has successfully modeled wireless network topologies by extensively using Poisson Point Process (PPP) \cite{baccelli-bartek}. The seminal work of Gupta-Kumar \cite{GuptaKumar} has enriched the community knowledge on the achievable limits of capacity. Further works (e.g., \cite{Bartek,mobihoc_to_cite_7}) have focused on the routing and communication properties of these topologies. Yet the Poisson Point Process is not adequate for the modeling of vehicles, as the repartition of cars is not uniform on the entire space but limited to streets. 

Getting one step closer to the reality of vehicular network has been done by the authors of \cite{vehicular_cox} that model the streets as a Poisson line tessellation and the position of cars on the streets as Poisson points on these Poisson lines. This work assumes uniform density of the nodes in each street.

Vehicular networks have specific requirements and challenges, as shown in \cite{mobihoc_to_cite_2}. In making use of the delay tolerant property of ah-hoc networks, an important work has been done in \cite{Jacquet:DTN}, where a thorough analysis is provided for the broadcast time in a delay tolerant network (DTN). Information propagation has been studied in \cite{cite1}, where the authors show that worm epidemics in VANET has
an initial linear growth rate which is much slower than the
exponential growth predicted by classical epidemiological models,
and observed in worm attacks on the Internet. The work presented in  \cite{cite2}, provides upper and lower bounds for
message propagation as function of traffic density, vehicle speed
and radio range, demonstrating
that increased mobility of vehicles actually aids in messaging.
Interesting observations were presented in \cite{cite_3}, for example that a message can propagate in the opposite direction as the vehicle
traffic flow and can propagate much faster than vehicle movement. The authors of \cite{cite_4} quantify network performance in terms of its ability to disseminate tracking information. \cite{cite5} provides an analysis of the information
propagation speed in bidirectional vehicular delay tolerant
networks on highways. In \cite{cite6}, the authors study the information propagation
process in a 1-D mobile ad hoc network formed by vehicles Poissonly
distributed on a highway and traveling in the same direction
at randomly distributed speeds that are independent between vehicles, while \cite{cite_n} proposes and evaluates
 models of the temporal connectivity
of vehicular networks.
Temporal connectivity has also been analyzed in \cite{cite_n_1} by exploiting a dataset of taxi GPS.

However, these models are either in the extremes of regularity, either randomness, or focus on particular communication setups (highways, intersections). 
Up to our knowledge, there is a lack of macro-models that merge the particular communication setups into a bigger picture that capture the environment's self-similarity. 

 Self-similarity has been 
 recently introduced for the topology of network nodes \cite{mascot1}.
Our hyperfractal topology has been recently introduced in \cite{spaswin} and \cite{gsi} where they are only limited to the static case, while this paper studies the dynamic aspects. Furthermore, one needs to prove that KPIs specific to wireless networks can be obtained with tractable expressions for the proposed model.
The construction model and, most importantly, the wireless propagation in cities lead to a Delay Tolerant Network. The network is intermittently interconnected as the spatial repartition of the mobile nodes and the mobility may force packets to ``wait'' for a vehicle to arrive and receive the packet. 

\section{System Model} \label{model}
\subsection{Hierarchical structure of traffic on roads}

Cities are hierarchically organized \cite{Batty2008TheSS}. The centers which form this hierarchy have many elements in common in functional terms and repeat themselves across several spatial scaling. In this sense, districts of different sizes at different levels in the hierarchy have a similar structure. The growth of cities not only occurs through the addition of units of development at the most basic scale, but through increasing specialization of key centers, thus raising their importance in the hierarchy.

In any urban environment there exists a hierarchy of streets based on their importance in the urbanization scheme 
 (e.g. boulevards, streets, alleys). The level that a street occupies in the urbanization scheme comes as consequence of the traffic density it supports. 
 An interesting yet intuitive observation is that the cumulated length of a street type decreases with the importance of the street. For example, the cumulated length of boulevards is lower that the cumulated length of streets, which, in turn, is lower than the cumulated length of alleys.

This phenomenon has been successfully captured in the work done in \cite{Gloaguen2006,gloaguen,gloaguen2} where the authors show that the fixed networks follow road networks and these are deployed in a hierarchical way. In these studies, the authors use iterated tessellations in order to capture this phenomenon.

 It comes as natural, therefore, that the average traffic density per cumulated length of each street type follows a power law of scaling.   
 In the following, we will analyze a particular case, when the power law is according to a fractal distribution.

\subsection{Hyperfractals}

Lauwerier et al.~\cite{LauwerierK} defined a fractal as a geometrical figure that consists of an identical motif repeating itself on an ever-reduced scale. Cities are  hie\-rarchically organized in self-similar structures and represent good candidates for being modeled using fractal geometry. 

To build a fractal, Mandelbrot starts with a geometric object called an initiator. To this he applies a motif which repeats itself at every scale calling this the generator. The fractal is obtained by applying the generator to the initiator, deriving a geometric object which can be considered to be composed of several initiators at the next level of hierarchy or scale down. Applying the generator again at the new scale results in further elaboration of the object's geometry at yet a finer scale, and the process is thus continued indefinitely towards the limit. In practice, the iteration stops at a level below which further scaled copies of the original object are no longer relevant for the purpose of the modeling.
In essence, however, the true fractal only exists in the limit, and thus what one sees is simply an approximation to it. 
Let us emphasis that while the computations performed throughout this work allow for the fractals to go to infinity, the simulations are performed for limited levels of construction.

We propose to use our model called hyperfractal, a model focused on the self-similarity of the topology. The hyperfractal model is not in the extremes of the regularity or randomness.  
Let us emphasize the fact that our definition of hyperfractal has no common elements with the definition given in \cite{hyperfractals_physics}.


The population is built in a finite window. The map is assumed to be the unit square.

The support of the population is a grid of streets
but with an infinite resolution.
The construction of the support reminds of a space filling curve. An example is displayed in Figure \ref{fig:map_support}. In the first stage the lines forming level 0 are drawn in thick black. In the second stage, each of the four areas obtained is again considered as an independent map with a specific scaling and the lines of level 1 are drawn in thinner black. The process is further continued in a similar manner in the third stage, where each of the 16 areas are again split by 16 crossed drawn in very thin black lines and the procedure continues. 

Let us denote this structure by $\mathcal{X}=\bigcup_{l=0}^\infty \mathcal{X}_l$ with 
\begin{align*}
\mathcal{X}_l:=&\{(b2^{-(l+1)},y), b=1,3,\ldots, 2^{l+1}-1, y \in [0,1]\} \\
\cup &\{(x,b2^{-(l+1)}), b=1,3,\ldots,2^{l+1}-1, x \in [0,1]\},
\end{align*}
where $l$ denotes the level and $l$ starts from $0$, and $b$ is an odd integer.
Figure~\ref{fig:map_support} displays, in fact, the levels for $l=0,1,2$. 
Observe that the central "cross" $\mathcal{X}_0$  splits $\bigcup_{l=1}^\infty \mathcal{X}_l$ in $4$ "quadrants" which all are  homothetic  to  $\mathcal{X}$ with the scaling factor $1/2$.

\subsection{Hyperfractal Mobile Node Distribution} \label{distribution}
The hyperfractal model for static networks has been already studied in previous works \cite{spaswin,gsi}. For completeness we briefly remind the reader the model before extending it with mobile nodes. 

The total population in the map is $n$ mobile nodes.
The process of assigning points to the lines forming the support is performed
recursively, in iterations, similar to the process of obtaining the Cantor Dust \cite{mandelbrot}.
We consider the Poisson point process $\Phi$ of (mobile) users on $\mathcal{X}$
with total intensity (mean number of points) $n$ ($0<n<\infty$)
having 1-dimensional intensity 
\beq
\lambda_l=n(p/2)(q/2)^l
\label{eq:dens_mobiles}
\eeq
on $\mathcal{X}_l$, $l=0,\ldots,\infty$,
with $q=1-p$ for some parameter $p$ ($0\le p\le 1$).
Note that $\Phi$ 
can be constructed in the following way: one samples the total  number of mobiles users $\Phi(\mathcal{X})=n$ from Poisson$(n)$ distribution; 
each mobile is placed independently 
with probability $p$ on $\mathcal{X}_0$ according to the uniform distribution
and with probability $q/4$ it is recursively located in the similar way in one the four quadrants of $\bigcup_{l=1}^\infty \mathcal{X}_l$.
Obviously  the process $\Phi$ is neither stationary nor  isotropic.
However it has the following self-similarity property:
the intensity  measure of  $\Phi$ on  $\mathcal{X}$ is hypothetically reproduced in each of the four quadrants of $\bigcup_{l=1}^\infty \mathcal{X}_l$ with the scaling of its support by the factor 1/2 and of its value by $q/4$. One can therefore define  the fractal dimension~\cite{fractal_dimension} of this measure. 
 \begin{remark}
 The fractal dimension $ d_F$ of the intensity measure of $\Phi$ satisfies
\begin{equation*} \label{eq:d_f}
\left(\frac{1}{2}\right)^{d_F}=\frac{q}{4} \qquad\text{thus}\qquad d_F=\frac{\log(\frac{4}{q})}{\log 2}\ge 2.
\end{equation*}
\end{remark}
Thus the measure has a structure which recalls the structure of a fractal set, such as the Cantor map~\cite{mandelbrot}. A crucial difference lies in the fact that the fractal dimension here, $d_F$, is in fact {\it greater} than 2, the Euclidean dimension, this is why our model is called a ``hyperfractal".
Notice that when $p=1$ the model reduces to the Poisson process on the central cross  while for $p\to 0$, $d_F\to 2$ it corresponds to  the uniform measure in the unit square.
  
Figure \ref{fig:hyper} shows the population obtained in the street assignment process after $4$ iterations. As one can easily notice, the population density decays with the street level.  
\begin{figure}  [ht!]
\begin{subfigure}[t]{0.225\textwidth}
\includegraphics[scale=0.17]{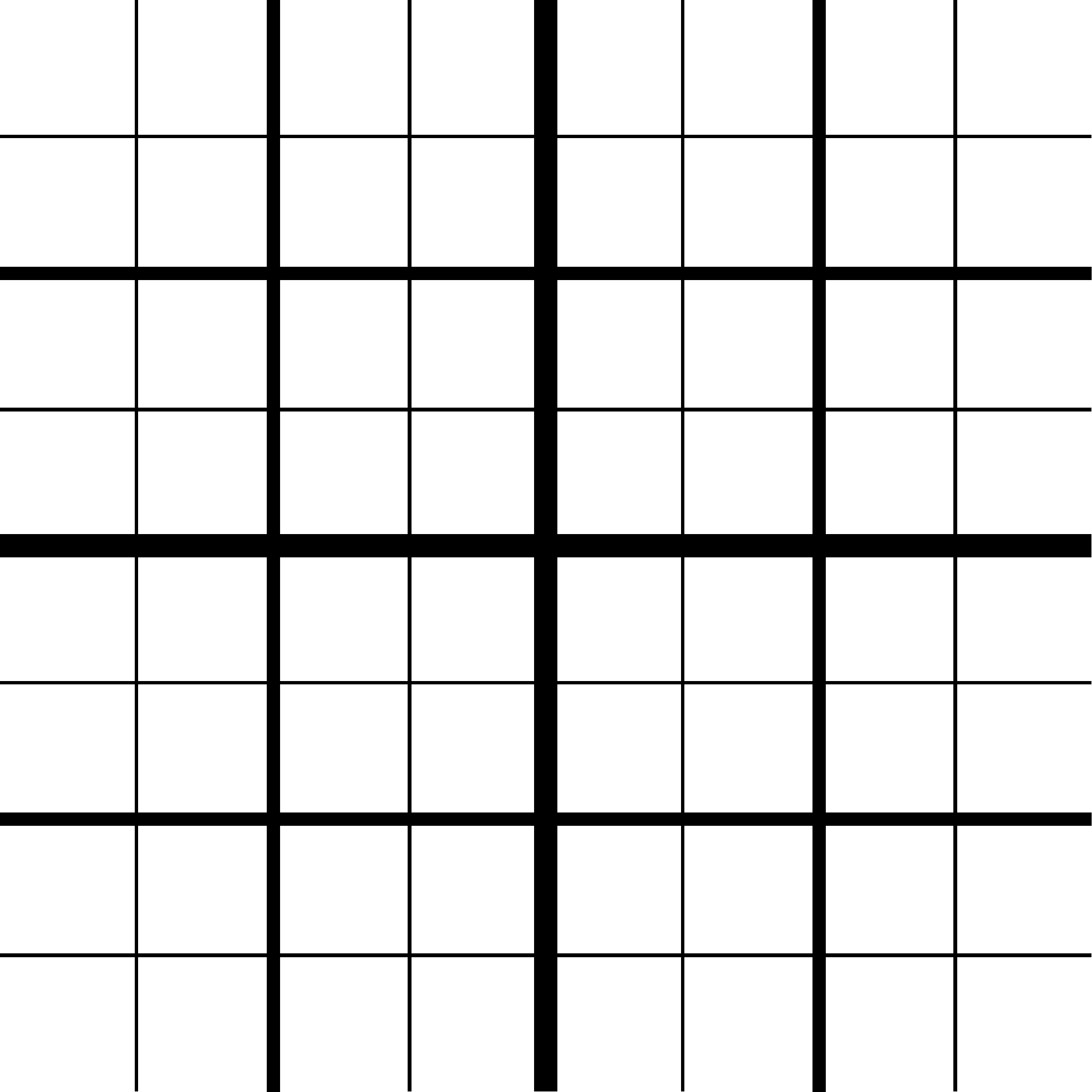}
\caption{}
\label{fig:map_support}
\end{subfigure}
\begin{subfigure}[t]{0.225\textwidth}
\includegraphics[scale=0.33,  trim={3cm 9.5cm 0cm 10cm}]{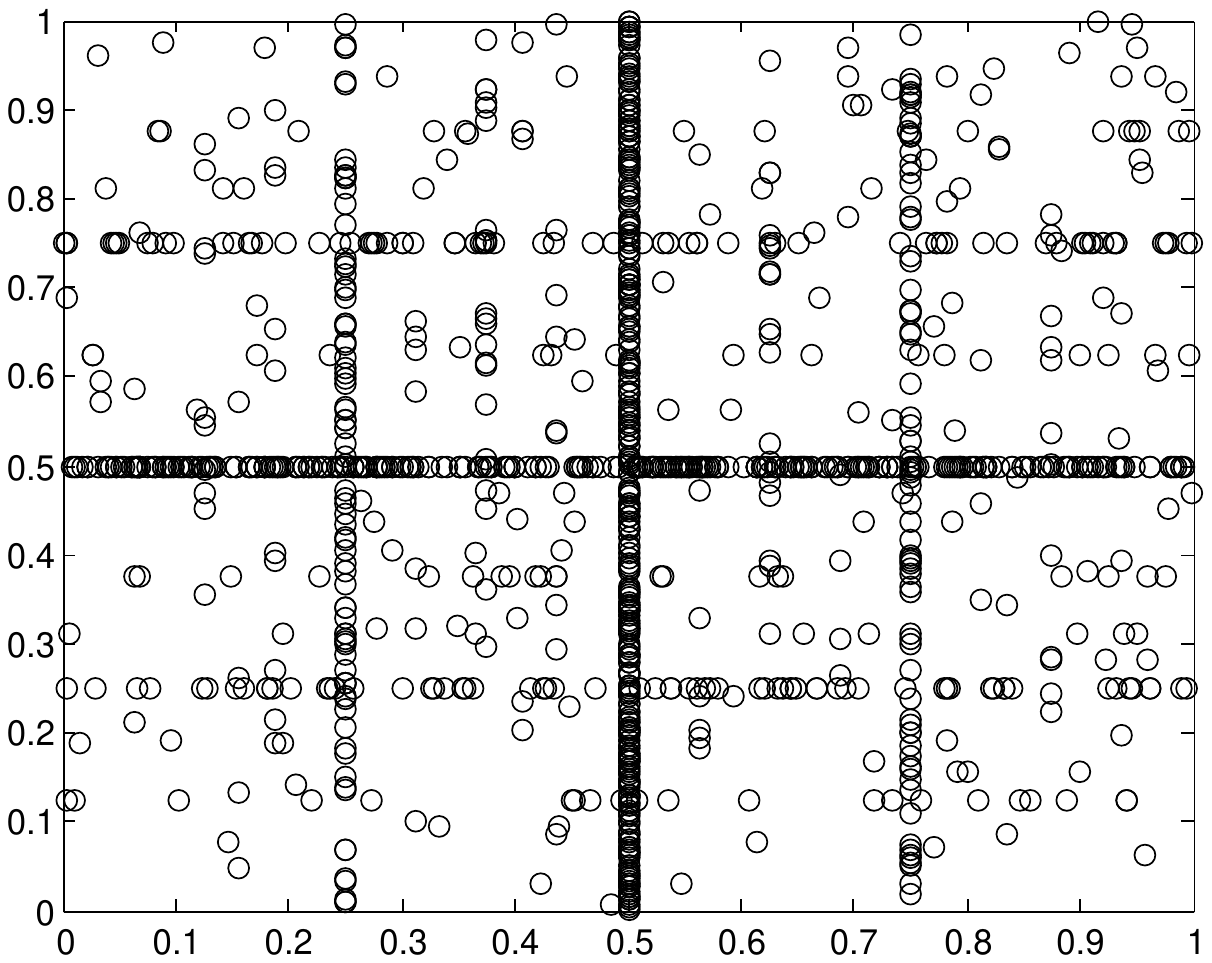}
\vspace*{-0.35cm}
\caption{}
\label{fig:hyper}
\end{subfigure}
\caption{(a) Hyperfractal map support; (b) Hyperfractal, $d_F=3$, $n=1,200$ nodes;}
\end{figure}

As the density of the population on the streets rapidly decays with the increase of the level, there will be unpopulated streets. Therefore, we say that a street is \textit{busy} if the street contains at least one mobile node. The following lemma will play an important role in the proof of our main results: 
\begin{lemma}\it\label{lembusy}
The average number of busy streets is asymptotically equivalent to $n^\delta\frac{-\Gamma(-\delta)}{\log(2/q)}$ with $\delta=\frac{\log 2}{\log(2/q)}$ 
where $\Gamma$ stands for Euler ``Gamma" function.
\end{lemma}
As mentioned, when $d_F\to 2$ the node distribution tends to a uniform Poisson distribution, and, in this case, the number of busy streets tends to $n$, since in an uniform Poisson point process all nodes are east-west or north-south aligned with probability zero, therefore each node stays on a separate street that contains just the node himself. This is well reflected in the result above since one would obtain $\delta\to 1$.

The mobiles move on the lines that are the support structure $\mathcal{X}$ of the hyperfractal nodes. 
When a node reaches a boundary, it reenters the map from the same point, following a billiard mobility.  

Initially, for the sake of simplicity, the 
speed of the mobiles is considered to be constant and identical, $v$, no matter the level and the density of the nodes on the lines. In reality, the values of speed vary in certain intervals. As our analysis is focused on upper and lower bounds, the variation of speed will not impact the order of magnitude of the broadcast time. The case of variable speed will be discussed later, in Section~\ref{simulations} where we show that the bounds are validated for variable speed case as well. The case where some streets are congested is a particular case that speeds up the broadcast since some nodes will be blocked in intersections. Due to space limit, we do not analyze this case here. Note that Ho {\em et al.}~\cite{Ho2011} show that the dependency inherited from vehicular interactions and traffic signals tends to be ignored by the system, and that the Poisson property still holds under such approximation of vehicle interactions (the distributional result of their model is validated against real-world empirical data in London).

\subsection{Canyon Effect}
The analyzed network modeled is an urban vehicular network that can be made realistic by adding an assumption capturing the reality of a well-known  phenomenon: the canyon effect \cite{canyon_model,urban_canyon1}. Buildings are made of concrete, glass and steel which generate a formidable obstacle for radio wave propagation. The {\it canyon} propagation model implies that the signal emitted by a mobile node propagates only on the street where it stands on. If the network was static, considering the given construction process, the probability that a mobile node is placed in an intersection goes to zero when the street width goes to zero and nodes positioned on two different streets are never able to communicate. Notice that when a street has positive width, the intersection width is negligible compared to the street length and the network will still be partitioned.  
The connectivity of the network is thus ensured through the mobility of the nodes, leading to a scenario of a {\it delay tolerant network}. 

\subsection{Broadcast algorithm}
A feasible approach to forward a packet of information from a source to a destination in the absence of any predictive knowledge on the node movement is the epidemic routing, analog to the spreading of an infectious disease. In this case, when the traffic is low, epidemic routing can achieve an optimal delivery delay at the expense of increased use of network resources.
The considered broadcast protocol is a single-hop broadcast meaning that each vehicle carries the information while traveling, and this information is transmitted to the other vehicles in its one-hop vicinity during the next broadcasting cycle. This single-hop broadcasting protocol relies heavily on the mobility of the vehicles for spreading information. In a first phase, one hop vicinity is considered to be represented only by nearest neighbors of the infected node.

We initially assume a constant average hop duration $h$ as we consider a constant communication individual load for all nodes (i.e., avoiding initially the case of local overload) to be able to compute a meaningful average broadcast time (from all possible initiating nodes). 

In this paper, 
as we primarily seek to understand the limit of the propagation speed, we do not consider other detailed aspects of the broadcast protocol, such as packet collisions. 
At time $t_0=0$ only one node, called ``source", holds the packet. At time $t>t_0$, the population of nodes is split among  nodes that have received the packet, called infected nodes, and nodes that have not yet received the packet,  called healthy nodes (by analogy with  epidemic propagations). 


\section{Main Results}
\label{results}

In this section we provide the computation of broadcast time when the network is modeled using a hyperfractal model. We show that the computations are simplified due to the scaling effect of the hyperfractal. In particular, given the self-repeating pattern of streets with hierarchical density, we can compute metrics of interest by observing a local scenario (for such a pattern).

The main results are first proven under the assumption that each node is reachable through wireless propagation by its nearest neighbor, therefore that the radio range is always high enough to reach the next hop. 
The results provided are as following: the evaluation of the generic upper and lower bounds for the average broadcast time in a hyperfractal setup. Then, specific results in extremes cases are provided. It will be shown that the performance is due in part to an interesting self-similar phenomenon, denoted as {\em information teleportation}, that arises as a consequence of the topology and allows an acceleration of the broadcast therefore decreasing the broadcast time.
We then provide the extension of the results when radio range is considered.

Throughout the following analysis, without lack of generality, we only consider streets which are busy streets as per Section~\ref{distribution}.

\subsection{Upper Bound}
There are interesting observations to be made on the hyperfractal model. These observations will lead to an intuitive computation of the bounds.
For example, the following remark is a consequence of the construction process.

{\em Remark:}
There are $2^H$ streets of level $H$ intersecting each of the streets forming the central cross.
Due to the canyon effect, the packet will not be able to jump from a street to another street, but has to be propagated through intersections. This will be done when a node carrying the packet crosses the intersection. 
We denote by $I(n_i,n_j)$ the average time that a packet takes to jump from one street containing $n_i$ nodes to an intersecting street containing $n_j$ nodes, assuming all nodes on the first street carry the packet. In fact, it is sufficient to assume that the closest nodes towards the intersection carry the packet. This quantity depends on the mobility pattern of the mobile nodes which will be detailed before Lemma~\ref{lem Inm}. 

In the following, for sake of generality, we deal with the case where the nodes of interest, $x$ and $y$, are placed on perpendicular lines of respective depths, $H_a$ and $H_b$. 
The location of $x$ denotes the location of the node on line of level $H_a$ that initiates the broadcast. 
The location of $y$ denotes the location of a node on line $H_b$ that will receive the packet. In fact, due to the mobility of the nodes, one cannot fix from the start of the broadcast the location of node $y$. Furthermore, as the time when the packet arrives from one route on a specific location on the line $H_b$ can differ from the case the packet comes from a different route, we cannot choose the location of a node in a fix moment in time.
Therefore, we choose $y$ to be the representation of a location of a node $y$ on line $H_b$ on the segment $[y_0-1/\lambda_{H_b},y_0+1/\lambda_{H_b}]$, where $y_0$ denotes the position of a node on line $H_b$ when the broadcast was initiated by node $x$.

\begin{definition}
We define by $T_n (x,y)$ the time necessary for a packet transmitted in a broadcast initiated by node $x$ to arrive at node $y$. We define by $\E[T_n (x,y)]$ the average broadcast time between all fixed $(x,y)$ pairs. 
\end{definition}
\begin{definition}
The direct route is the route that uses the streets that embed the nodes $x$ and respectively, $y$ and contains the intersection between these two streets.
A diverted route between nodes $x$ and $y$ is a route that employs four segments and three intersections.  
\end{definition}

As an example, in Figure \ref{fig:scenario}, the direct route is drawn in red dotted line and the diverted route is drawn in continuous red line and continuous blue line.  
\begin{figure}\centering
\includegraphics[scale=0.3, trim={11cm 4cm 9cm 1.5cm}]{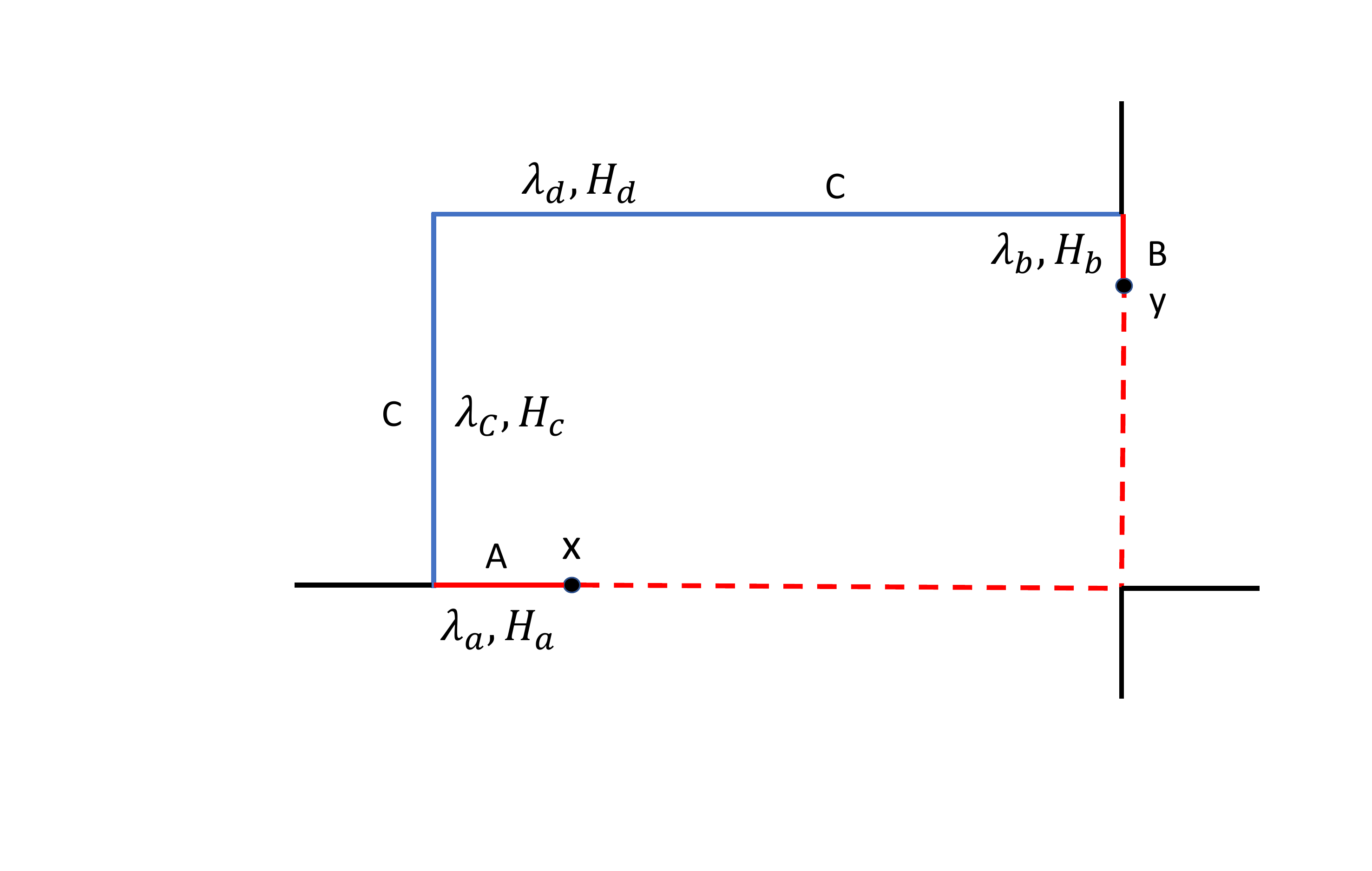}		
	\vspace{-0.0cm}
	\caption{Direct route versus directed route}
	\label{fig:scenario}
\end{figure}

\begin{lemma}\it\label{upper1}
Assume that the street of level $H_a$ holding node $x$ contains $n_a$ nodes, and the street of level $H_b$ holding node $y$ contains $n_b$ nodes. Given a transmission hop time $h$, 
the bound on the time $T_n(x,y)$ in the direct route is:
\begin{equation}
\label{eq:mineq1}
T_n(x,y)\le hn_a +hn_b+I(n_a,n_b).
\end{equation}
\end{lemma}
\begin{proof}
To prove the formulae, let us look at the setup in Figure \ref{fig:scenario}. In the direct route case it is assumed that the packet hops from node $x$ towards the intersection with the street holding the node $y$ (straight, blue line). The maximum number of hops is $n_a$, thus it takes at most $hn_a$ time units. The packet turns on the intersection in time $I(n_a,n_b)$ and then proceeds towards node $y$ in at most $h n_b$ time units. 
\end{proof}
\begin{lemma}
Consider the diverted route containing two additional streets of level $H_c$ and level $H_d$, respectively perpendicular to the street holding $x$, and perpendicular to the street holding $y$, and containing respectively $n_c$ nodes and $n_d$ nodes. Assuming $n_c$ and $n_d$ strictly positive, in the diverted route, the bound on $T_n(x,y)$ becomes:
\begin{eqnarray}\label{upper2}
T_n(x,y)&\le & h(n_aL(H_c)+1)+h(n_b L(H_d)+1)
\nonumber\\
&&+I(n_a,n_c)+I(n_b,n_d)+I(n_c,n_d)
\nonumber\\
&&+hn_c+hn_d
\label{eq:mineq}\end{eqnarray}
where $L(H_i)$ is the distance from a node to the intersection with a street of level $H_i$.
\end{lemma}
		
\begin{proof}
Let us again look at the setup illustrated in Figure \ref{fig:scenario}.
It is assumed that, instead of taking the direct route as expressed in inequality (\ref{eq:mineq1}), the packet is diverted into the street of level $H_c$, then to the street of level $H_d$ before being delivered on the street of node $y$ (dashed, red line). The quantity $h(n_a L(H_c)+1$ and $h(n_bL(H_d)+1)$ is the time necessary for the packet to propagate from $x$ to the intersection with the street of level $H_c$ through hop-by-hop propagation, and similarly, the time necessary for the packet to propagate from the intersection of the street of level $H_d$ to node $y$. The average number of nodes between $x$ and the streets of level $H_c$ is $n_a L(H_c)$, due to the uniform node distribution on the interval. Even in considering the node mobility and the time difference between each hop, the mean remains the same since the distribution of nodes remains uniform on the interval.
The additive term $+1$ in the final result comes from the fact that the closest node moving towards the intersection may be located beyond the intersection.
\end{proof}

The following lemma gives an estimate of the packet turn time at an intersection.
\begin{lemma}
For all $n_i$, $n_j$ $\in$ $\mathbb{N^*}$, the following inequality holds:
\begin{equation}
\label{eq:hoping}I(n_i,n_j)\le\frac{1}{v(n_i+n_j)}.
\end{equation}
\label{lem Inm}\end{lemma}
\begin{proof}

Let $\Delta$ be the distance between an intersection and a node moving toward the respective intersection. It is assumed, without loss of generality (we assume that all roads are bidirectional), that the packet originally progresses on an East-West street and that the intersection stands at abscissa $z\in[0,1]$.

Let $g>0$ be the distance toward the intersection of the closest node moving toward the intersection. An upper bound on the probability that $\Delta>g$: $\Delta$ is greater than $g$ when no node is either in the interval $[z-g,z]$ and is not moving toward the right, or is in the interval $[z,z+g]$ and is moving toward the left. The probability of such event is upper bounded by $(1-g)^n$, in fact it is exactly this expression when $z\in[g,1-g]$, omitting border effects. 

The car at distance $\Delta$ reaches the intersection in $\frac{\Delta}{v}$ time units. At this time the car can transmit the packet to the closest car on the North-South street. Now, merging the problem over the two streets together, the probability that the time for the packet to turn to be larger than $g/v$ is equal to the probability that no car on the East-West street and on the North-South street reaches the intersection before $g/v$ time units, which is upper bounded by $(1-g)^{n_i}(1-g)^{n_j}$. Thus
$$
vI(n_i,n_j)\le\int_0^1(1-g)^{n_i+n_j}dg=\frac{1}{n_i+n_j}.
$$
\end{proof}
The following technical result gives the probability that a street is busy and will be further used in the development of the bounds.
\begin{lemma}\label{lemma_pos}
In a hyperfractal with $n$ nodes and $d_F>2$, the probability that a street a level $H(n)$, with $n_c$ mobile nodes on the street and  $H(n)=\lceil \frac{\log(n^{1-\epsilon}p/2)}{\log(2/q)} \rceil$ is empty is smaller than $e^{-(q/2)n^\epsilon}$ .
\end{lemma}

\begin{proof}
\begin{eqnarray*}
\mathbb{P}\left(n_c=0\right)=\left(1-\lambda_{H(n)}\right)^n
<e^{-n\lambda_{H(n)}}\leq e^{-(q/2) n^{\epsilon}}
\end{eqnarray*}
\end{proof}
The following theorem gives the upper bound on the broadcast time and proves that it grows as $n^{1-\frac{1}{d_F-1}}$, showing that the growth is sub-linear and only depends on the fractal dimension. 


\begin{theorem}\it\label{toinf}
Consider a network with $n$ mobile nodes in a hyperfractal setup with fractal dimension $d_F>2$, transmission hop time $h$ and constant speed of vehicles $v$. Let $x$ and $y$ be two nodes on perpendicular streets. When $n \to \infty$ the average broadcast time satisfies:
\begin{itemize}
\item{(i)} in the direct route scenario
\begin{equation}\label{direct1}
\E[T_n(x,y)]\le hpn+\frac{1}{v};
\end{equation}
\item{(ii)} in the diverted route scenario, for all $\epsilon>0$ 
\begin{eqnarray}\label{second_upper}\nonumber
\E[T_n(x,y)]\le & 2hp n^{1-\delta+\epsilon}\left(\frac{2}{p}\right)^\delta+\frac{4}{q}hn^{\epsilon/\delta}\\
&+\frac{3}{v}+O(ne^{-(q/2)n^{\epsilon}})
\end{eqnarray}
\end{itemize}
where $\delta=\frac{1}{d_F-1}$. 
\end{theorem}
\begin{proof}


As we look for the upper bound, the inequality in the direct case (Eq.~\ref{direct1}) comes straightforward as $\E[n_a+n_b]<pn$. The term $\frac{1}{v}$ is the upper bound of $I(n_i,n_j)$ therefore maximizing the sum.    
The diverted route scenario again follows  Figure \ref{fig:scenario}. 
Both $n_c$ and $n_d$ are strictly positive with high probability. 
Indeed the probability that $n_c=0$ is $e^{-n\lambda_c}$. Let us take $H_c=H_d=H(n)$ with $H(n)=\lceil \frac{\log(n^{1-\epsilon}p/2)}{\log(2/q)} \rceil$ as per Lemma \ref{lemma_pos}, thus: 

$$
\frac{\log(n^{1-\epsilon}(p/2))}{\log(2/q)}-1 \le H(n)\le \frac{\log(n^{1-\epsilon}(p/2))}{\log(2/q)}+1. 
$$
We have both:
$$
\E[n_c]=\E[n_d]\le n^\epsilon \frac{2}{q}.
$$

Meanwhile, let us take as diverted route the closest street of level $H(n)$ from node $x$ since this street is busy with probability higher than $e^{-(q/2)n^{\epsilon}}$. In this case:
$$
L(H(n))\le 2^{-H(n)}\le 2n^{-\delta(1-\epsilon)}\left(\frac{2}{p}\right)^\delta.
$$
Consequently: 
$$
\frac{p}{2}n2^{-H(n)}\le pn^{1-\delta+\delta\epsilon}\left(\frac{2}{p}\right)^\delta
$$
and the result is obtained by changing the value of $\epsilon$ in $\epsilon/\delta$.

The term $O(ne^{-(q/2)n^{\epsilon/\delta}})$ comes from the case when either $n_c=0$ or $n_d=0$ which arrives with probability $e^{-(q/2)n^\epsilon}$. In this case, we know that $T_n(x,y)\leq hpn+\frac{1}{v}$ thus the contribution to $\E[T_n(x,y)]$ is $O(ne^{-(q/2)n^{\epsilon}})$.
\end{proof}
\paragraph{Remark} As $\epsilon$ becomes smaller, the convergence of eq. (\ref{second_upper}) is slower. 
\paragraph{Remark} The quantity $\delta$ is strictly less than 1 ($\delta<1$) and tends to 1 when $d_F\to 2$.
\paragraph{Remark} The term $3/v$ in Equation~\ref{second_upper} is far too high to upper bound the three turns, as the meaningful value for each turn comes from the term $\frac{1}{v(n_i+n_j)}$ in Equation~\ref{eq:hoping}.  Also its weight in the sum is not be of high importance. Using Lemma \ref{eq:hoping}, the term can be replaced by $\frac{2}{v n}+\frac{1}{vn^{\epsilon/\delta}}$. Notice that the optimal value of $\epsilon$ is of order $\frac{\log\log n}{\log n}$ which does not reach negligible values as long as $n\ll\frac{1}{hv}$ and leads to the global estimate $\E[T_n(x,y)]=O( n^{1-\delta \log n})$.


\begin{definition}
The average broadcast time $T_{broadcast}$ is the average of all source-destination pairs $(x,y)$ of $\E[T_n(x,y)]$. 
\end{definition}

\begin{corollary}\label{coro1}
When $n \to \infty$, the average broadcast time as the average over all sources $x$ satisfies:
\beq
T_{broadcast}=O(n^{1-\delta}\log n)
\eeq
\end{corollary}

\subsection{Lower Bound}
The intuition behind the lower bound of the average broadcast time comes from the fact that the highest weight in the broadcast time is taken by the time that the packet hops on the main cross, where the density of mobile nodes is considerably higher than on the following levels.
\begin{theorem}\it\label{lower}
Consider a network with $n$ mobile nodes in a hyperfractal setup with fractal dimension $d_F>2$, transmission hop time $h$, constant speed of vehicles $v$, $\delta=\frac{1}{d_F-1}$ and $\Gamma$ stands for Euler ``Gamma" function. The average broadcast time satisfies:
\begin{equation}
T_{broadcast}\ge \frac{p^3}{2} h n^{1-\delta}\frac{\log(2/q)}{-\Gamma(-\delta)}
\end{equation}
\end{theorem}
\begin{proof}
The broadcast time verifies:
\beq \label{eq:all_terms}
T_{broadcast}=  \frac{1} {n^2} \sum_{(x,y)}\E[T_n(x,y)]
\eeq
when $(x,y)$ are all the possible pairs of two nodes in the hyperfractal.

We denote by $\CH$ the set of nodes on the horizontal segment belonging to the central cross and $\CV$ the set of nodes on the vertical segment of the central cross. 

As we compute the sum only over the terms on the central cross, 
\beq \label{eq:central_terms}
T_{broadcast}\geq  \frac{1} {n^2}\left( \sum_{\substack{x \in \CH  \\ y\in \CV }} \E[T_n(x,y)]+\sum_{\substack{x \in \CV  \\ y\in \CH }} \E[T_n(x,y)]\right)
\eeq
as the number of the terms in the sum in (\ref{eq:central_terms}) is lower than the total number of terms in the sum in (\ref{eq:all_terms}). 

Since the packet must leave the street of node $x$, it must at least run on a distance $L(x)$ which is the average distance from node $x$ to the closest busy perpendicular street. The same holds for reaching node $y$. In other words, the following inequality holds:
\begin{equation}\nonumber
\E[T_n(x,y)]>hpn~\E[L(x)].
\end{equation}
Assume that the node $x$ is on the East-West segment of the central cross. The average distance to the closest North-South busy street is larger than $\frac{1}{2{\bf NS}_n}$, where ${\bf NS}_n$ is the random variable expressing the number of busy North-South streets in presence of $n$ mobile nodes. Therefore, $E[L(x)] \geq \E[\frac{1}{2{\bf NS}_n}]$. The lowest value would be obtained if the busy North-South streets were equally spaced. 

By Lemma~\ref{lembusy}, one has $E[{\bf NS}_n]=\frac{1}{2}B_n$.
Furthermore $E\left[\frac{1}{{\bf NS}_n}\right]\ge \frac{1}{E[{\bf NS}_n]}$ by convexity of the hyperbolic function. Thus by referring to Lemma~\ref{lembusy}: 
$$
\E[T_n(x,y)]>\frac{hpn}{B_n}=hpn\frac{\log(2/q)}{-\Gamma(-\delta)}.
$$
Using the fact that $\E[|\CH||\CV|]=(n-1)np^2/4$ terminates the proof.
\end{proof}

\begin{corollary}\label{coro2}
The average broadcast time when $n \to \infty$ satisfies:
\beq
T_{broadcast}=\Omega(n^{1-\delta})
\eeq
\end{corollary}

Remark: When $n=2$, following expression (\ref{direct1}), $T_{broadcast}<h+\frac{1}{v}$. Furthermore, the same holds for the cases when all nodes are on the same street or move on two perpendicular streets.

Combining Corollary \ref{coro1} and the lower bound of Theorem~\ref{lower}, one obtains the matching result:

\begin{corollary}\label{coro3}
The average broadcast time when $n \to \infty$ satisfies:
\beq
\lim_{n \to \infty} \frac{\log T_{broadcast}}{\log n}=1-\delta
\eeq
\end{corollary}

\subsection{Asymptotic to Poisson Uniform}

As previously mentioned, the asymptotic case when $\delta=1$ gives a Poisson uniform case. This scenario works as follows: each mobile node is placed randomly on the plane and moves on one of the two possible motion directions: North-South or South-North (resp., East-West or West-Est). Note that all streets are bidirectional. Every node is alone on its road, the only occasion when a car can communicate is when another car crosses its road. E.g., a single node moving on a East-West street sees and transmits a packet towards all the North-South (or South-North) moving nodes crossing its street in $O(\frac{1}{v}+h)$ time. Furthermore, one of these nodes moving on a North-South street transmits the packet to all East-West moving nodes also in $O(\frac{1}{v}+h)$ time, thus the total broadcast time is:
\beq
T_{broadcast}=O(\frac{1}{v}+h)
\eeq
which is equivalent into letting $\delta=1$ in the general formula.

\subsection{Extension with limited radio range}

When a car correctly receives the packet, it transfers it to all the cars that are withing its radio range. Up until this moment, throughout this work, we have used the hypothesis of unlimited radio range. In this section we will investigate the more realistic hypothesis of limited radio range.

In the following, the radio range is dependent of the number of mobile nodes in the city map, $R_n=\frac{1}{\sqrt{n}}$. The reason is the following.  
The population of a city (in most of the cases) is proportional to the area of the city and the population of cars is proportional to the population of the city, therefore the population of cars is proportional to the area of the city, $Area=A \cdot n$ where $A$ is a constant. A natural assumption is that the absolute radio range, $R$, is constant. But since we assume in our model that the city map is always a unit square, the relative radio range in the unit square must be $R_n=\frac{R}{\sqrt{An}}$ which we simplify in $R_n=\frac{1}{\sqrt{n}}$.

As the radio range is fixed and the average distance between nodes increases with the increase of the depth, some nodes will become unreachable. Therefore, the condition for a piece of information to be broadcasted on a street is that the average distance between nodes is not higher than the radio range. 

The following Lemma is an adaption of Lemma \ref{lemma_pos} and gives the maximum depth of the level on which the average distance between nodes allows for the propagation of the packet.


\begin{lemma}\label{lemma_pos_radio}
In a hyperfractal with $n$ nodes and $d_F>2$, the probability that a street a level $H(n)$ with    $H(n)=\lceil \frac{\log(n^{1/2-\epsilon}p/2)}{\log(2/q)} \rceil$ has at least one inter-node gap higher than $R_n$ is $ne^{-(q/2)n^\epsilon}$ .
\end{lemma}

\begin{proof}
The probability that a car is not followed by another car within distance $R_n$ is equal to $(1-R_n\lambda_{H(n)})^n$ which is smaller than $e^{-n^{1/2}\lambda_{H(n)}}$. Given $n_c$, the probability that there exists such a node (the car within distance $R_n$ of the car holding the packet) is smaller than $n_c e^{-\sqrt{n}\lambda_{H(n)}}$. With $n_c\le n$ the lemma is proved.
\end{proof}

Similarly, Lemma \ref{lembusy} becomes: 
\begin{lemma}\it\label{lembusyR}
The average number of busy streets where there is no inter-node gap higher than $R_n$ is asymptotically equivalent to $n^{\delta/2}\frac{-\Gamma(-\delta)}{\log(2/q)}$ with $\delta=\frac{\log 2}{\log(2/q)}$ 
where $\Gamma$ stands for Euler ``Gamma" function.
\end{lemma}
The proof follows the proof of Lemma \ref{lembusy}, by adding the factor $R_n$ to $\lambda_H$.

In this case, the upper bound rewrites as follows.
\begin{theorem}\it\label{toinf_radio}
Consider a network with $n$ mobile nodes in a hyperfractal setup with fractal dimension $d_F>2$, transmission hop time $h$ and constant speed of vehicles $v$. Let $x$ and $y$ be two nodes on perpendicular streets. When $n \to \infty$, for a transmission radio range of $R=\frac{1}{\sqrt{n}}$, the average broadcast time satisfies:
\begin{itemize}
\item{(i)} in the direct route scenario
\begin{equation*}\label{direct1_radio}
\E[T_n(x,y)]\le hpn+\frac{1}{v};
\end{equation*}
\item{(ii)} in the diverted route scenario, for all $\epsilon>0$ 
\begin{eqnarray*}\label{second_upper_radio}\nonumber
\E[T_n(x,y)]\le & 2hp n^{1/2-\delta+\epsilon}\left(\frac{2}{p}\right)^\delta+\frac{4}{q}hn^{\epsilon/\delta} \\
&+\frac{3}{v}+O(ne^{-(q/2)n^{\epsilon}})
\end{eqnarray*}
\end{itemize}
where $\delta=\frac{1}{d_F-1}$. 
\end{theorem}

The lower bound becomes:
\begin{theorem}\it\label{lower_radio}
Consider a network with $n$ mobile nodes in a hyperfractal setup with fractal dimension $d_F>2$, transmission hop time $h$, constant speed of vehicles $v$, $\delta=\frac{1}{d_F-1}$ and where $\Gamma$ stands for Euler ``Gamma" function. For a transmission radio range of $R=\frac{1}{\sqrt{n}}$, the average broadcast time satisfies:
\begin{equation*}
T_{broadcast}\ge \frac{p^3}{2} h n^{1-\delta/2}\frac{\log(2/q)}{-\Gamma(-\delta)}
\end{equation*}
\end{theorem}
 The proof follows the same steps as the proof of Theorem \ref{lower}, using the result of Lemma \ref{lembusyR} instead of Lemma \ref{lembusy}.

\subsection{Information Teleportation}\label{teleportation}

As Theorem \ref{toinf} shows, in a hyperfractal, the broadcasted packet can follow either a direct route or a diverted route. 
The diverted route case leads to the existence of new contagions on the lines of level $H_c$ and $H_d$. This is what we call ``information teleportation" phenomenon as the new contagions are not due to a source on lines $H_c$ or $H_d$ spreading its packet in a hop by hop manner but is due to  routing the packets through intersections. The phenomenon will be visually illustrated by experiments in  Section \ref{simulations}.

The teleportation phenomenon allows an acceleration of the broadcast time. Note that the acceleration itself is a self-similar phenomenon and takes places recursively: propagation on level $H_i$ is accelerated by teleportation coming from lines $H_{i+1}$, $H_{i+2}$, $H_{i+3}$ and so on. 
 In a hyperfractal with teleportation effect, the broadcast time evolves as $O(n^{1-\delta})$ according to Corollary \ref{coro1}.

 To consider a network with the absence of teleportation is to consider the direct route case in Theorem~\ref{toinf}. In such a network, the broadcast time scales linearly with the number of hops, $O(nh)$.
 The two regimes are illustrated in Figure \ref{fig:inflexion}.

\begin{figure}\centering 
\includegraphics[scale=0.5, trim=3cm 8.5cm 0cm 9cm]{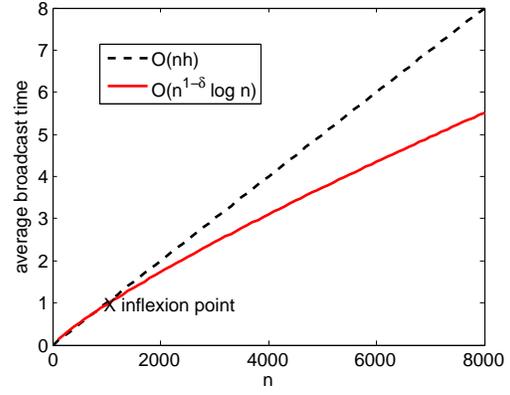}
\caption{Broadcast time evolution in hyperfractal vs linear regime: inflexion point.}
\label{fig:inflexion}
\end{figure}
\begin{figure}
\includegraphics[scale=0.5, trim=2cm 9.5cm 0cm 10.2cm]{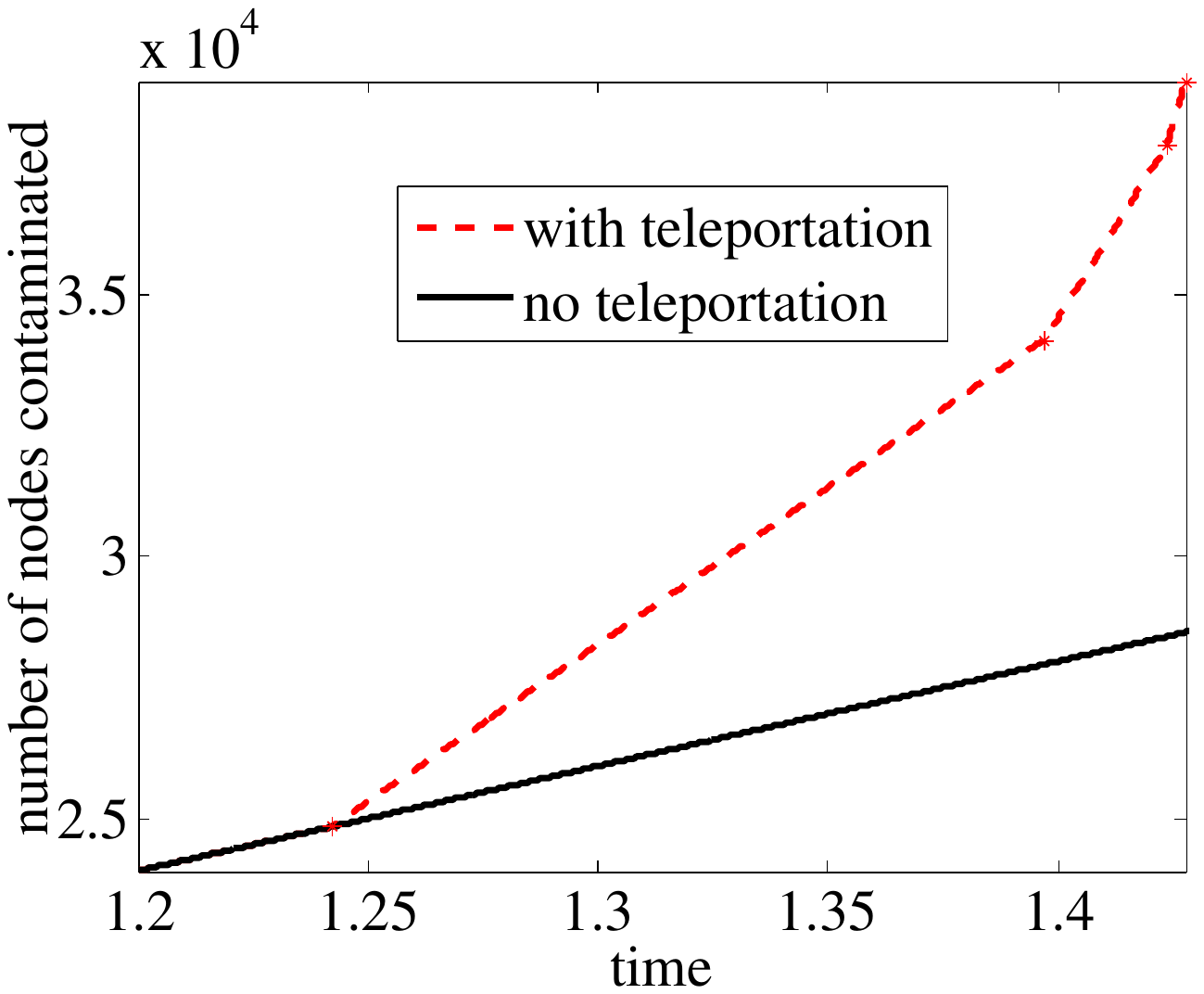}
\caption{Contagion propagation speed on the line of level $0$}
\label{fig:propagation_speed}
\end{figure}

The teleportation phenomenon arises after the linear characteristic overtakes the one for $O(n^{1-\delta})$. 
Therefore, the inflexion point where teleportation arises verifies $nh=An^{1-\delta}$, where $A$ is a constant. 

Let us look at an example of broadcast speed up that occurs due to information teleportation. In a hyperfractal setup, let us consider an infected source on a line of level $H=0$. In each time slot $h$, two more nodes get infected and become themselves sources. In the absence of teleportation, therefore, the number of infected points increases linearly with the hop time. 
In the presence of intersections with lines  of levels $H_i>0$, new ``outbreaks of infection'' arise at time stamps of $\frac{1}{vn(\lambda_0+\lambda_{H_i})}$. 

Figure \ref{fig:propagation_speed} shows graphically the evolution. This is an upper bound as not all the teleportation contagions generate a speed up; the nodes can be infected from neighboring contagions by simple hop by hop propagations.  
The teleportation phenomenon will be further analysed with the 
simulations in Section~\ref{simulations}.

\section{Model fitting to real cities traffic data}\label{calcul_df}
A mandatory requirement when providing a novel model for wireless networks is the development of a procedure that allows the transformation of data into the model with specific model parameters. 
Typical point process procedures of data fitting have been developed in the research community based on the different methods. For example, in R, a commonly used language by the stochastic geometry community, the functions allow fitting the points to several types of processes: Poisson, Strauss, Softcore, etc. 
Unfortunately existing procedures of data fitting cannot be used for the hyperfractal model as the interaction between points are different and cannot be recognized by existing software. 

To validate our model and prove its utility and ease of use, we developed a procedure of transforming traffic flow maps into hyperfractals, more precisely the computation of fractal dimension of the traffic flow maps. One can use such a procedure to compute a city/region fractal dimension and then compute metrics of interest. An example of such metrics is the broadcast time (see the previous section).

\subsection{Theoretical Foundation}

Let us emphasize that in the definition of the hyperfractal model, we did not make assumptions or conditions on geome\-tric properties such as shape. The model only needs density and length. For example, a hyperfractal does not need that either the main/first level streets to be in a cross or that there exist exactly two streets of level one that have the exact length. What is necessary is the scaling between the length of different levels of the support $\mathcal{X}_l$ and the scaling  of the 1-dimensional intensity per level, $\lambda_l$. 

Taking into acccount these observations (that come naturally from the construction process), we now elaborate a procedure of computation of the fractal dimension of a traffic density map. 
The procedure can be adapted by adding three criteria to increase the precision of the fitting: namely, density-to-length,  spatial intersection density, and time interval intersection.


\subsubsection{Density-to-length criteria and the computation of the fractal dimension}

This is the criteria used for computing the fractal dimension of the map. 
In a hyperfractal, the cumulated length of the street up to level $H$ is $2^{H+1}-1$. At this level, $H$, the density of the nodes on the streets is $\frac{p}{2} \left(\frac{q}{2}\right)^H$. Let us define the density as a function of the cumulated distance $\xi$. It can be expressed as:
\begin{equation*}
\lambda(\xi)=\Theta\left( \xi^{\log(q/2)/\log 2}\right)
\label{eq:fitting1}
\end{equation*}

Which can be further reduced to:
\beq \label{eq:eq_fitting}
\lambda(\xi)=\Theta\left(\xi^{1-d_F}\right)
\eeq
when $\xi$ increases, which is the indicator of the fractal dimension.

The \textbf{procedure for the computation of the fractal dimension} has the following four steps.  

(i) We first start by collecting the data, the length of streets and traffic statistics. For example, we used average annual traffic statistics. 

(ii) Next, we consider a single street as an alignment of consecutive segments whose densities, from the less dense segment to the densest segment, do not vary more than by a factor $A>1$. We call the density of the street the average density of its segment. In a pure hyperfractal city model $A=1$ similarly to the standard concept of quantization.

(iii) The following step is to rank the streets in decreasing order of density: $\lambda_1 \geq\lambda_2 \geq \ldots \geq \lambda_i \geq \ldots$ and to compute the vector of cumulated sums of the segments of streets ordered by their decreasing density. 

(iv) We next plot the density of sorted streets versus cumulated length of sorted streets. In parallel, we plot the density repartition function with a starting value of $d_F$ and by using the measure cumulated length and by curve fitting, determine the best approximation for $d_F$.

%


\subsubsection{The spatial intersection density criterion} This provides the density statistics of the street intersections in the map.  
As illustrated with the teleportation phenomenon, the propagation relies on the succession of crowded streets to less crowded streets with the possibility of routing  the packet from one street type to another. An accurate computation of the street intersection statistics is thus important for the validity of the computed metrics. 

For this, we must characterize the variation of the distance towards the intersection with a street whose density is in an interval $[a,b]$. More precisely we define $L([a,b])$ as the largest distance from any point in any street of $\CH$ (resp. $\CV$) to an intersection with a street in $\CV$ (resp. $\CH$) whose density is in the interval $[a,b]$.

There should exists $C>1$ such that
\begin{equation}
L([\lambda,C\lambda])=O\left(\lambda^{1/(d_F-1)} \right)
\end{equation}
when $\lambda$ decreases. 
In the pure hyperfractal model one must take $C=2/q$, otherwise some value of $\lambda$ would not correspond to any street density.

 \subsubsection{The time interval intersection criterion} This is only relevant for particular wireless metrics, like broadcast, but not all others. More specifically, the criteria is necessary for the validation of lemma \ref{lem Inm} that gives the estimate of the packet turn time at an intersection.

The average time interval $I(n,m)$ between two event crossing by mobile nodes at an intersection of two streets containing respectively $n$ and $m$ mobile nodes:
\beq
I(n,m)\le \frac{S}{n+m}
\eeq
where $S$ is a fixed parameter that relates to the average slowness of mobile nodes (defined informally as the average time to travel across one unit of distance). Formally, in our model, $S=1/v$, when considering constant speed.

\subsection{Data Fitting Examples}
To illustrate how the hyperfractal model can be used for representing vehicles distribution on streets, we present some data fitting results. 
Using public measurements \cite{dataMin}, 
\cite{dataNyon}, we show that the data validates the hyperfractal scaling of density and length of streets. While traffic data is becoming accessible, the exact length of each street  is particularly difficult to find.
This will become easier as real data and multiple synthetic large-scale simulators become readily available~\cite{cologne,luxembourg}.

\begin{figure}
\includegraphics[scale=0.5,  trim={2cm 8.5cm 0cm 9cm}]{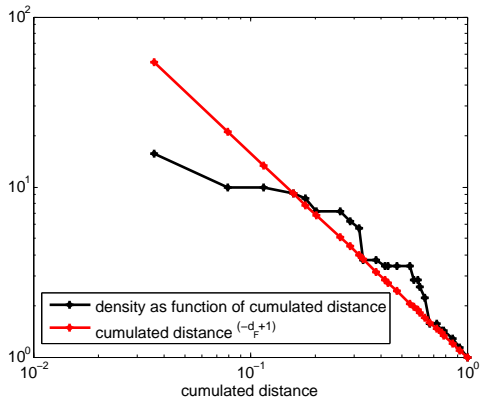}
\caption{Computation of fractal dimension of Seattle}
\label{fig:fitting1}
\end{figure}

Figure \ref{fig:data} shows the snapshot of a traffic flow map displaying the average annual weekday traffic in a neighborhood of Seattle. By applying the fitting procedure and using equation (\ref{eq:eq_fitting}), the estimated fractal dimension for Seattle is $d_F=2.3$.  
In Figures \ref{fig:fitting1} we show the fitting of the data for the density repartition function.  
Note that it is the asymptotic behavior of the plots that are of interest (i.e., the increasing cumulated distance with decreasing density) since the scaling property comes from the roads with low density, thus the convergence towards the rightmost part of the plot is of interest.
 
We provide a second example of data fitting using the measurements  of Minneapolis (see traffic map in the complementary document, with fitting from Adelaide and Nyon). By making use of the density-to-length criteria, the fractal dimension is computed to be $d_F=2.9$ and the fitting is further displayed in Figure \ref{fig:nyon_fit}.

\begin{figure}\centering
\includegraphics[scale=0.5,  trim={1cm 8cm 0cm 9cm}]{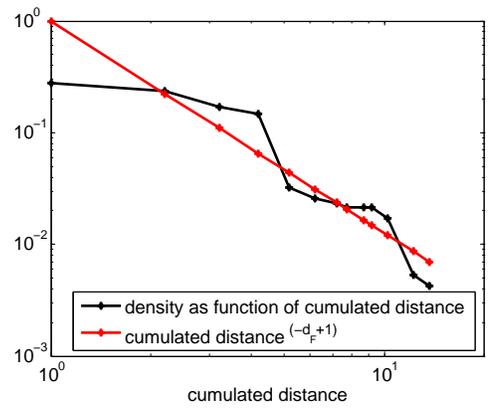}
\caption{Computation of fractal dimension of Minneapolis}
\label{fig:nyon_fit}
\end{figure}

The fitting procedure presented here allows for the computation of the fractal dimension of the map according to the available data. We do not suggest, however, that the annual average traffic flow is a good data set, as it does not capture variations such as day/night, rush hour/light traffic, and so on. An accurate modeling will need to adapt to the dynamics of the traffic measurements (and thus of the network), and the fractal dimension for each of these situations should be computed in each situation. In particular, the metrics of the network should be computed with the appropriate fractal dimension, depending on the situation analyzed (rush hour, night, etc). 

\section{Simulations}\label{simulations}

\subsection{Simulations in a system level simulator}
\subsubsection{QualNet Network Simulator Configuration}

In order to evaluate the accuracy of our theoretical findings, we have performed simulations in QualNet Developer 6.1 \cite{qualnet}. QualNet is a system level simulator capable to mimic the performance of real networks. 
The simulator allows the design and configuration of protocols, network topologies, propagation environment and traffic applications.

A custom configuration of the physical (PHY) and medium access control layer (MAC) layers allows the modeling of wireless networks implementing the IEEE 802.11p.
Table 1 presents the main PHY, MAC and propagation parameters configured for the simulated scenarios.

\begin{table}\centering
\begin{tabular}{|c|c|}
\hline
\textbf{Parameter}           & \textbf{Setting} \\ \hline
Modulation and coding scheme & 802.11 PHY specific MCS  \\ \hline
Operating frequency          & 2.4GHz         \\ \hline
Data rate              & 6Mbps      \\ \hline
Transmission power           & 20dBm           \\ \hline
Receiver sensitivity         & -85dBm               \\ \hline
Antenna type                 & Omnidirectional       \\ \hline
MAC Protocol            & 802.11 MAC     \\ \hline
Medium Access Technique & CSMA/CA       \\ \hline
Association mode        & Ad-hoc       \\ \hline
Street length &1 Km\\ \hline
Environment type &     urban   \\ \hline
 Application type &   Constant Bit Rate (CBR)    \\ \hline
 Hop duration time (h) &  60 ms  \\ \hline
\end{tabular}
\caption{PHY, MAC, application and environment parameters configuration in QualNet}
\end{table}\label{table1}

\subsubsection{Urban Vehicular Environment Modeling and Scenario Description}
An important step in the simulation is the design of the city map. As per the procedure described in Section \ref{calcul_df}, we can use the fractal dimension (computed from the input the average daily traffic flow measurement) to generate an equivalent simulated map of the city maps. 
%
The urban environment is replicated by means of a 3D map of a grid street plan, modeled using a three-level fractal geometry. We generate three levels in the hierarchy of streets (boulevards, streets, alleys) to set  the urban street grid.

As the buildings generate the canyon effect, the streets behave like a wave guide which is directly represented by the pathloss model set (urban model) and the propagation environment (metropolitan).
The width of the street decreases as the hyperfractal level increases, similarly to real cities, e.g., boulevards are wider than streets, which are, in turn, wider than alleys. For the first level the street width is approximately 60 meters, 30 meters for the second level, and 15 meters for the third level. 

The nodes are generated with a hyperfractal distribution, connected in a wireless ad-hoc network deployed in an urban environment (see Figure~\ref{fig:map_route}).
In release 6.1, QualNet does not offer broadcast capabilities. In order to simulate a single-hop broadcast protocol (i.e. the information carried by one vehicle is transmitted to the other vehicles in its one-hop vicinity during the next broadcasting cycle), static hop-by-hop routing was configured in QualNet, ensuring that packets are forwarded from one node to the one in the immediate vicinity. This allows to observe the propagation of the packet along the direct route versus diverted route, as studied in 
Section \ref{results}.

We first validate the bounds claimed by Theorems \ref{toinf} and \ref{lower} by observing the time necessary for a packet to propagate from one fixed source S to a fixed area of location, by selecting a destination node D.
As such, there are two considered end-to-end paths between the source node S and the destination node D: the direct route (in Figure \ref{fig:map_route} in red) which follows the first-level streets (with a higher density of nodes), and a diverted route (in Figure \ref{fig:map_route} in green) which uses third-level streets (with a lower density).

\begin{figure}\centering
\includegraphics[scale=0.65, trim={0cm 0cm 0cm 0cm}]{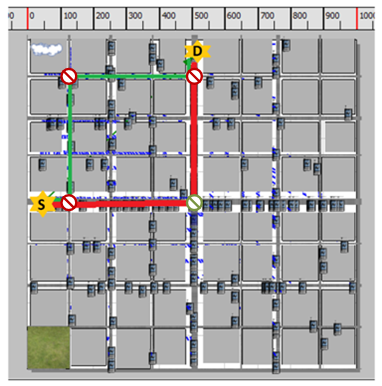}
\caption{Direct (green) versus diverted (red) path between source and destination, $d_F=3$, $n=200$}
\vspace*{-0.3cm}
\label{fig:map_route}
\end{figure}

An important part of the end-to-end delay is the time spent in routing the packet through the intersections. As QualNet does not support delay-tolerant features, the store-carry-and-forward phase was emulated as following. When a node X is in the ``store and carry" phase, the time the packet should spend in the node's memory may be higher than the packet expiration time defined in the protocol. For the packed not to be dropped, we pass it to a ``virtual" node placed in the intersection. When the node X reached the intersection, we allow the ``virtual" node to forward its packet, thus emulating a store-carry-and-forward process. As such, additional nodes were added in intersection and were temporarily deactivated during the simulation. 

\subsubsection{Validation of upper and lower bounds: constant speed}

Several batches of simulations were run for three values of fractal dimension: $d_F=2.5$, $d_F=3$, and $d_F=3.75$ with the number of nodes $n$ ranging from $n=200$ up to $n=800$ nodes. The end-to-end delay on the considered paths is evaluated by using a Constant Bit Rate (CBR) application generating $100$ packets of $512$ bytes at every $5$ seconds.


The formulations used for the upper bounds are the expression in equation (\ref{direct1}) for the direct route and (\ref{second_upper}) for the diverted route respectively.
For the lower bound the formulation used for validation through simulations is the closed expression $T>hnL(x)$.

Figure \ref{fig:bounds_constant} validates Theorems \ref{toinf} and \ref{lower} on the expression of the average broadcast time for the direct path. The speed of the mobile nodes has been set to $40$ kmph, as the typical legal speed limit in many cities.
The upper bound is depicted in dash black, the lower bound in dash blue and the simulation results in continuous red line.

\begin{figure*}\label{fig:bounds_constant}
\begin{subfigure}[b]{0.325\textwidth}
\includegraphics[scale=0.43, trim=4cm 9cm 0cm 10cm]{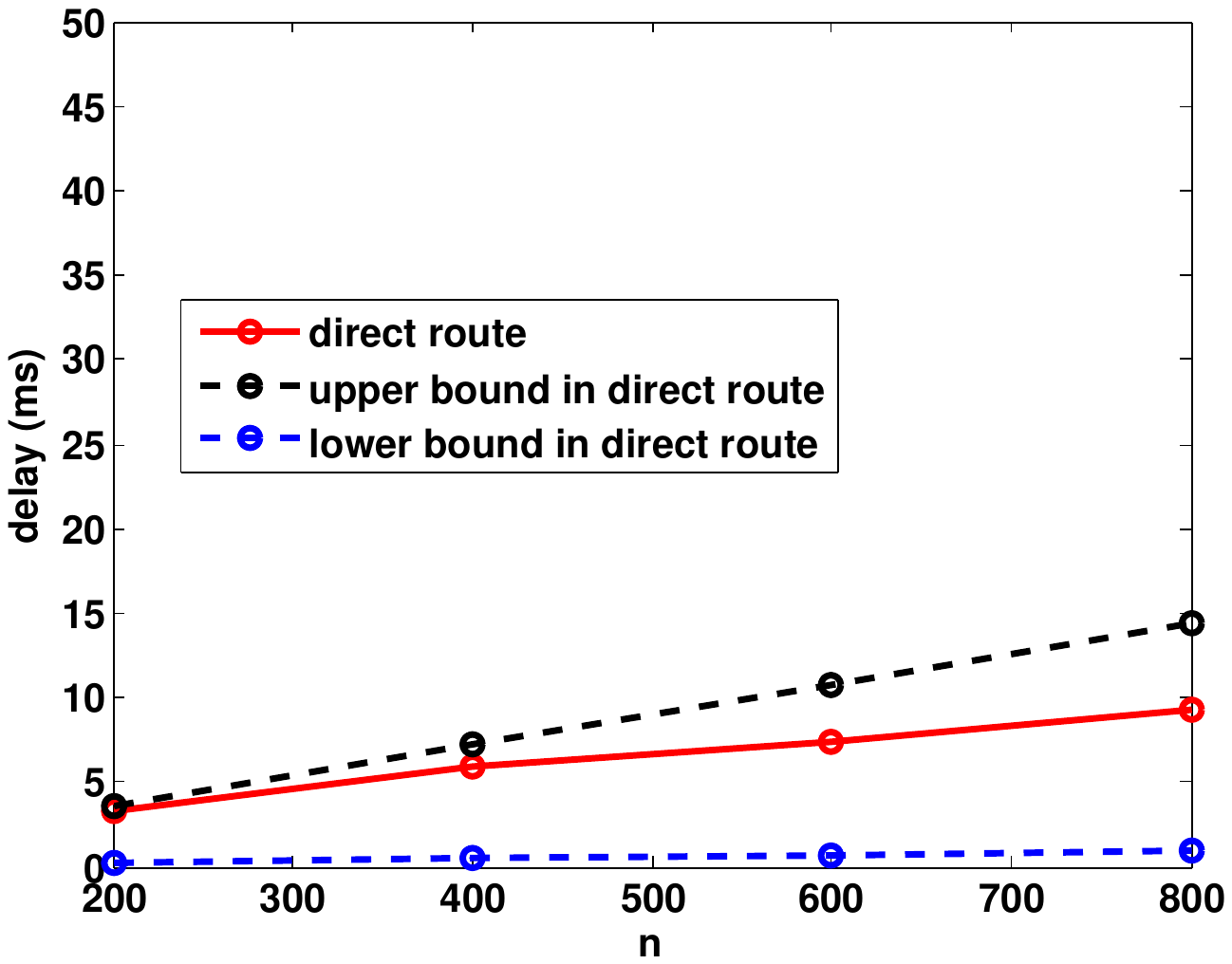}
\caption{$d_F=2.5$}
\vspace*{-0.12cm}
\end{subfigure}
\begin{subfigure}[b]{0.325\textwidth}
\includegraphics[scale=0.43, trim=4cm 9cm 0cm 10cm]{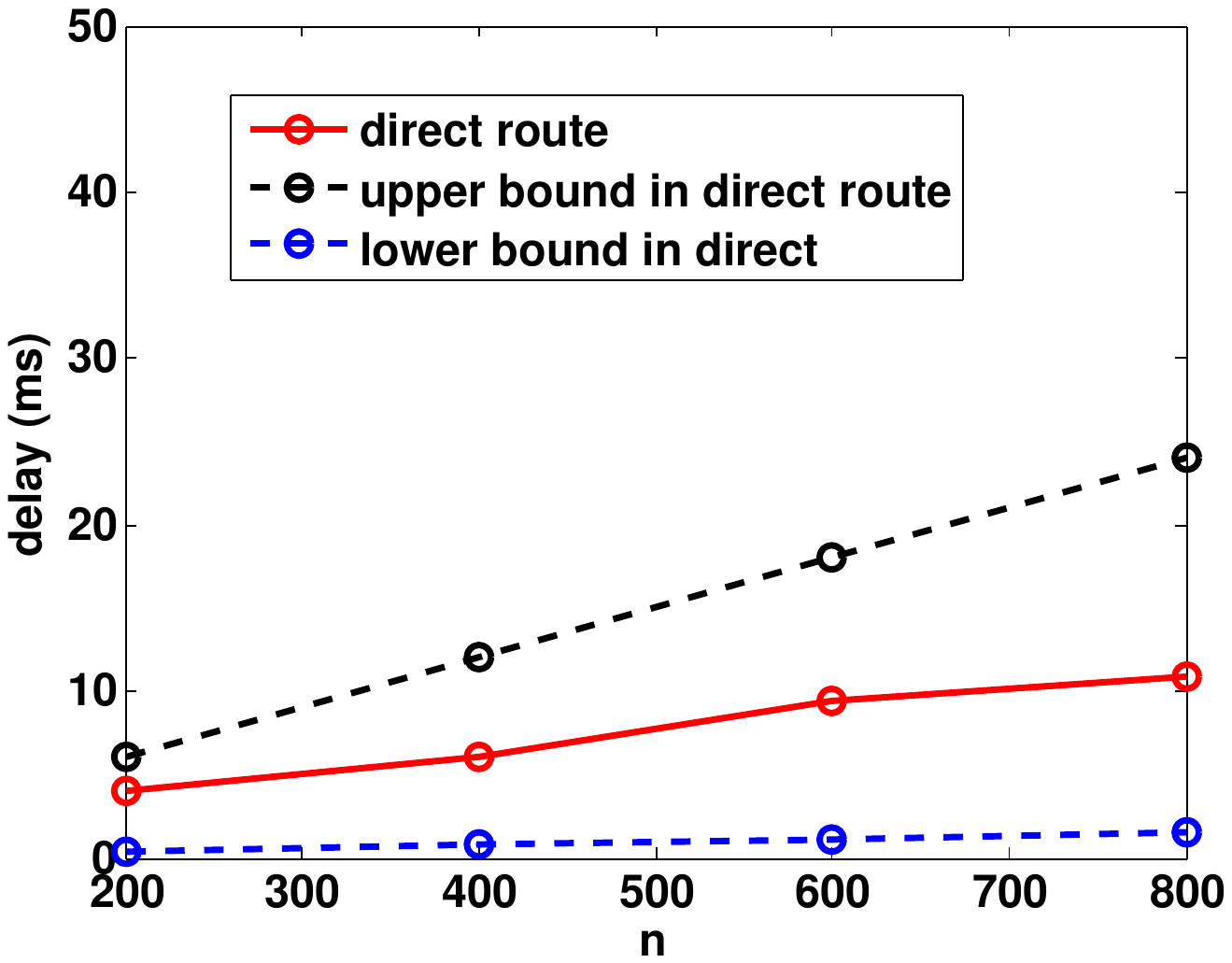}
\caption{$d_F=3$}
\end{subfigure}
\begin{subfigure}[b]{0.325\textwidth}
\includegraphics[scale=0.43, trim=4cm 9cm 17cm 10cm]{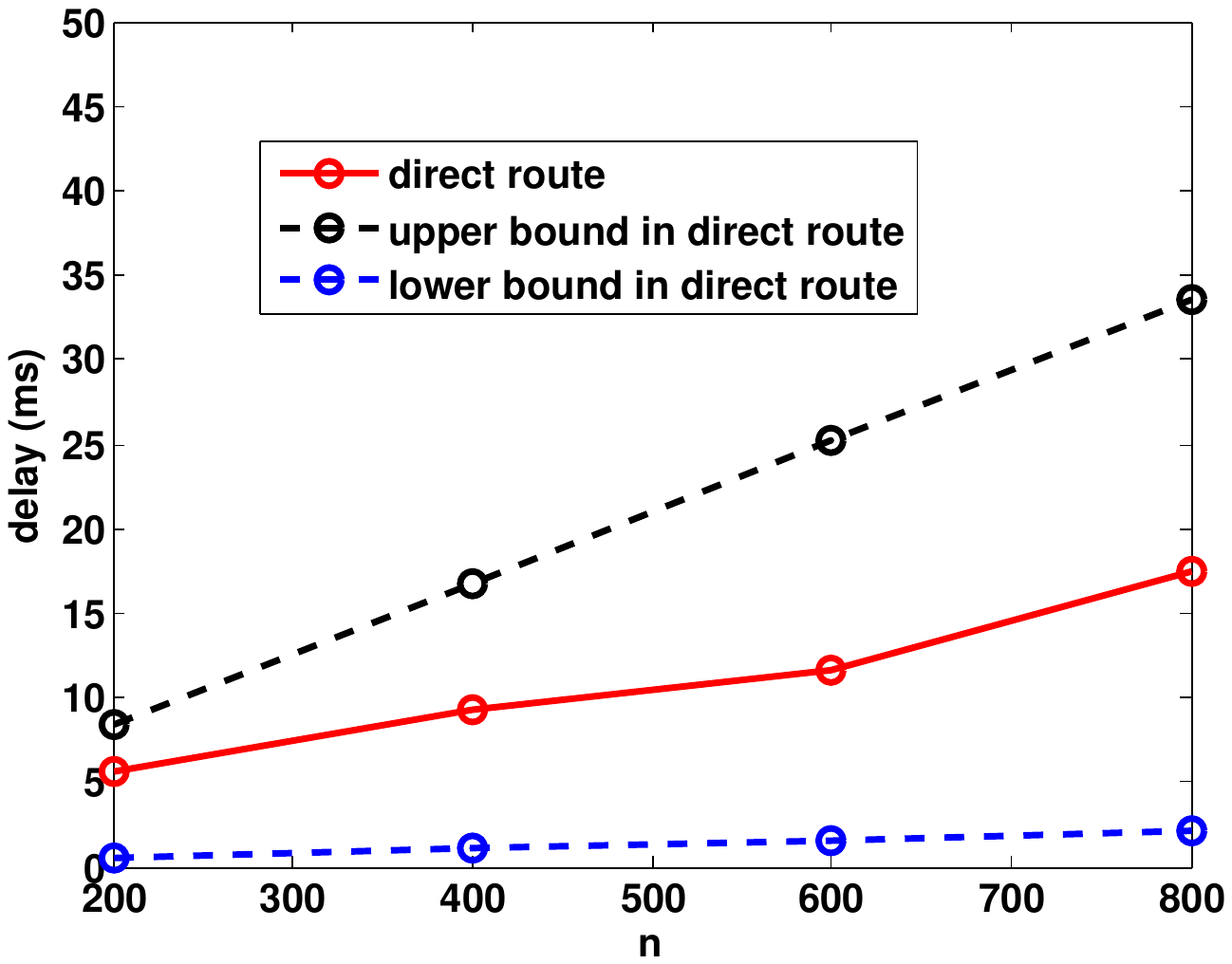}
\caption{$d_F=3.75$}
\end{subfigure}
\caption{Validation of bounds for direct route, constant speed}
\label{fig:bounds_constant}
\end{figure*}
Figure \ref{fig:bounds_diverted} validates Theorems \ref{toinf} and \ref{lower} on the expression of the average broadcast time for the diverted path. Again, the speed of the mobile nodes has been set to the constant value of $40$ kmph.  
\begin{figure*}[h!]\label{fig:bounds_diverted}
\begin{subfigure}[b]{0.325\textwidth}
\includegraphics[scale=0.43, trim=4cm 9cm 0cm 10cm]{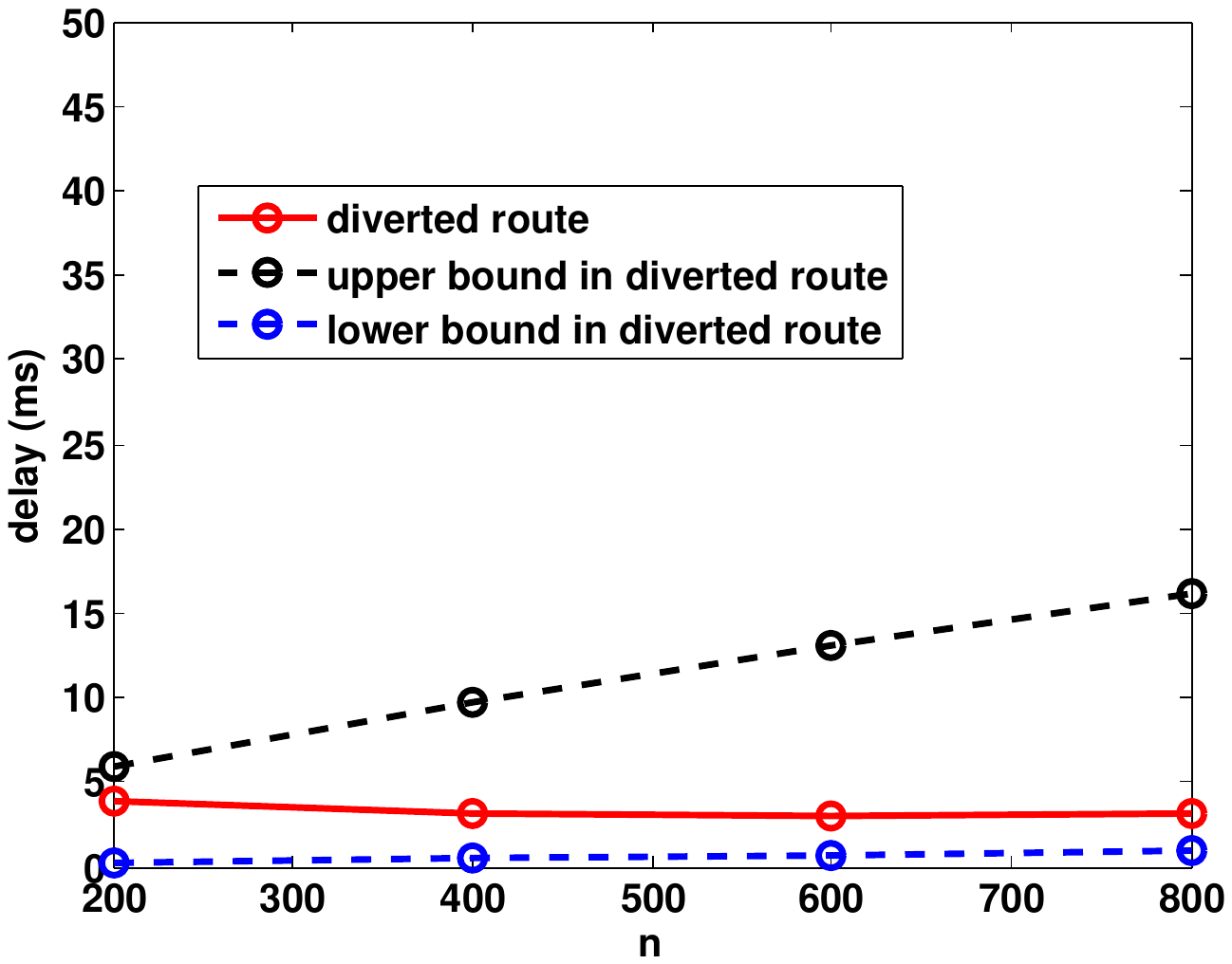}
\caption{$d_F=2.5$}
\vspace*{-0.12cm}
\end{subfigure}
\begin{subfigure}[b]{0.325\textwidth}
\includegraphics[scale=0.43, trim=4cm 9cm 0cm 10cm]{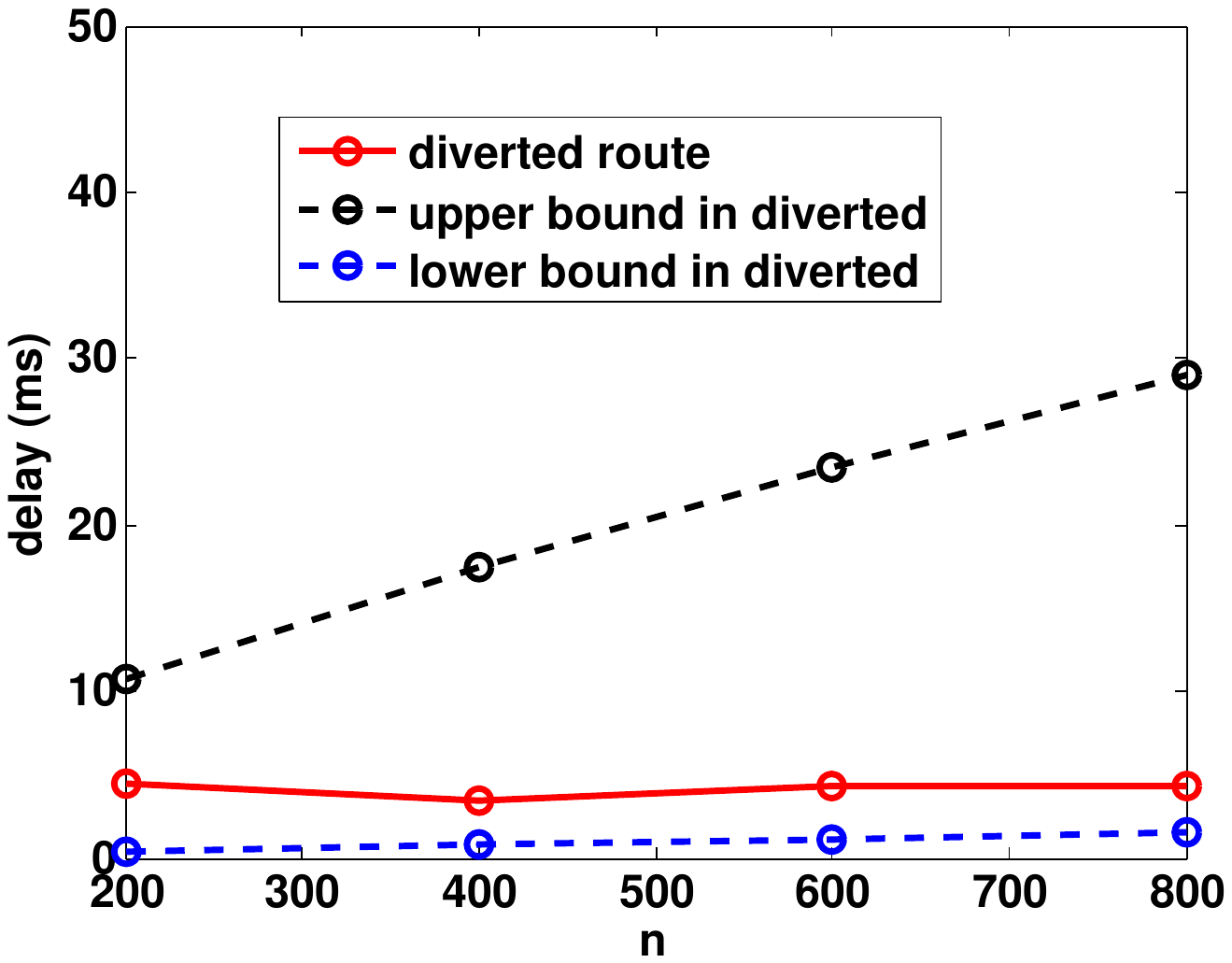}
\caption{$d_F=3$}
\end{subfigure}
\begin{subfigure}[b]{0.325\textwidth}
\includegraphics[scale=0.43, trim=4cm 9cm 17cm 10cm]{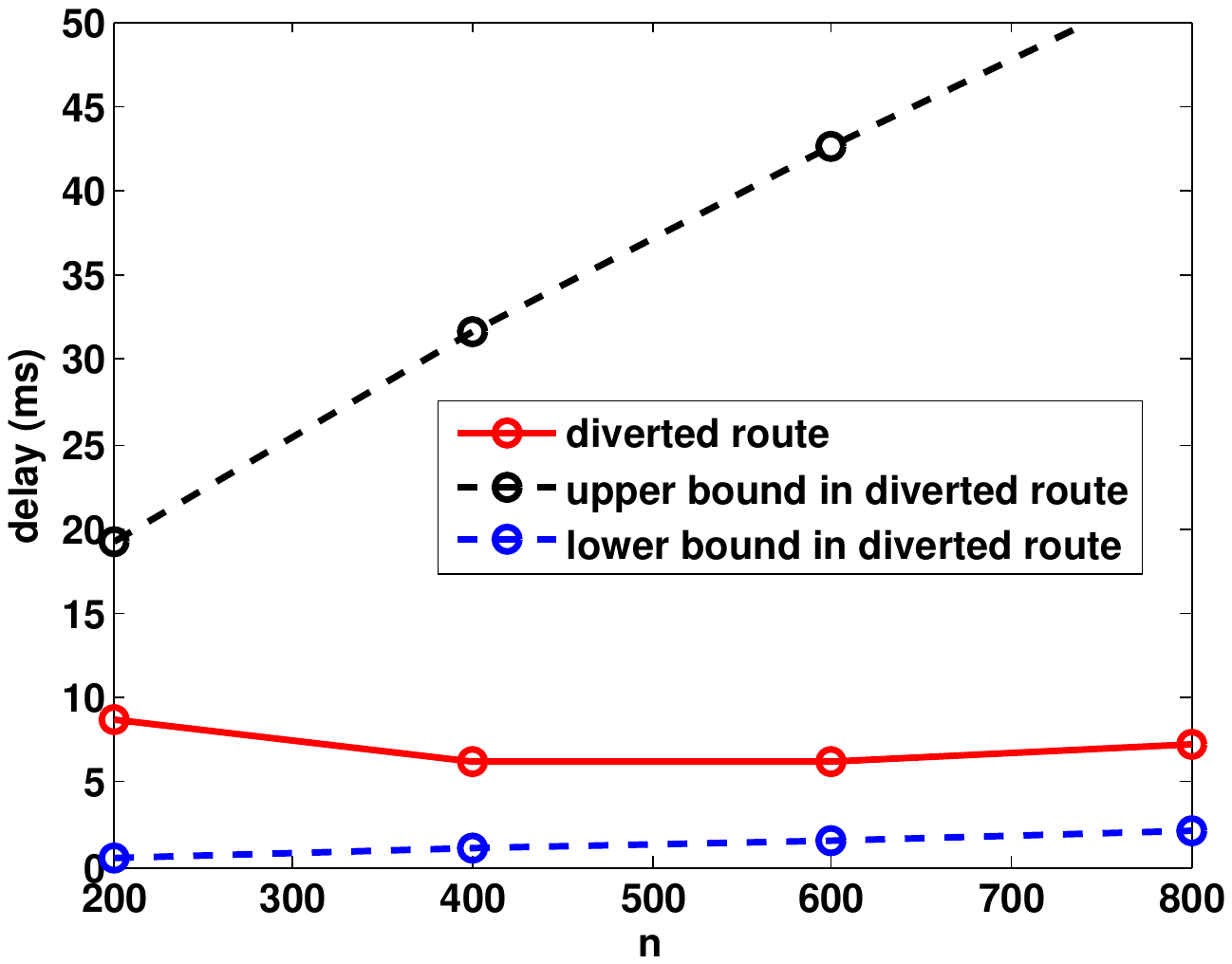}
\caption{$d_F=3.75$}
\end{subfigure}
\caption{Validation of bounds for diverted route, constant speed}
\label{fig:bounds_diverted}
\end{figure*}

\subsubsection{Validation of bounds under speed variation}
The theoretical results are developed under the assumption of constant speed throughout the whole network map, yet research has shown \cite{Straight} that the speed is highly influenced by the environment geometry. The following experiments show that the bounds introduced in Section \ref{results} hold for the more realistic scenario of variable values of speed.

The first analysis looks at the case where the speed of the nodes is proportional to the level, i.e., $v \propto H$. This models the scenario where the speed is lower on crowded streets, due to congestion, and increases with the decrease of density of nodes. Namely, the nodes on a level $H_i$ have the same speed $v_i\propto H_i$.  
For each simulation scenario, the following vehicular mobility scenario is considered: the vehicles move with $20$ kmph on the level 0 streets, $40$ kmph on level 1 and $60$ kmph on level 2, therefore the speed on the red route is $20$ kmph and on the green route $60$ kmph. Figure \ref{fig:speed_direct} validates the bounds, showing that our results extend to variable speed cases.

\begin{figure} [hb!]
\begin{subfigure}[b]{0.225\textwidth}		\includegraphics[scale=0.35, trim={4.3cm 8cm 9cm 9cm}]{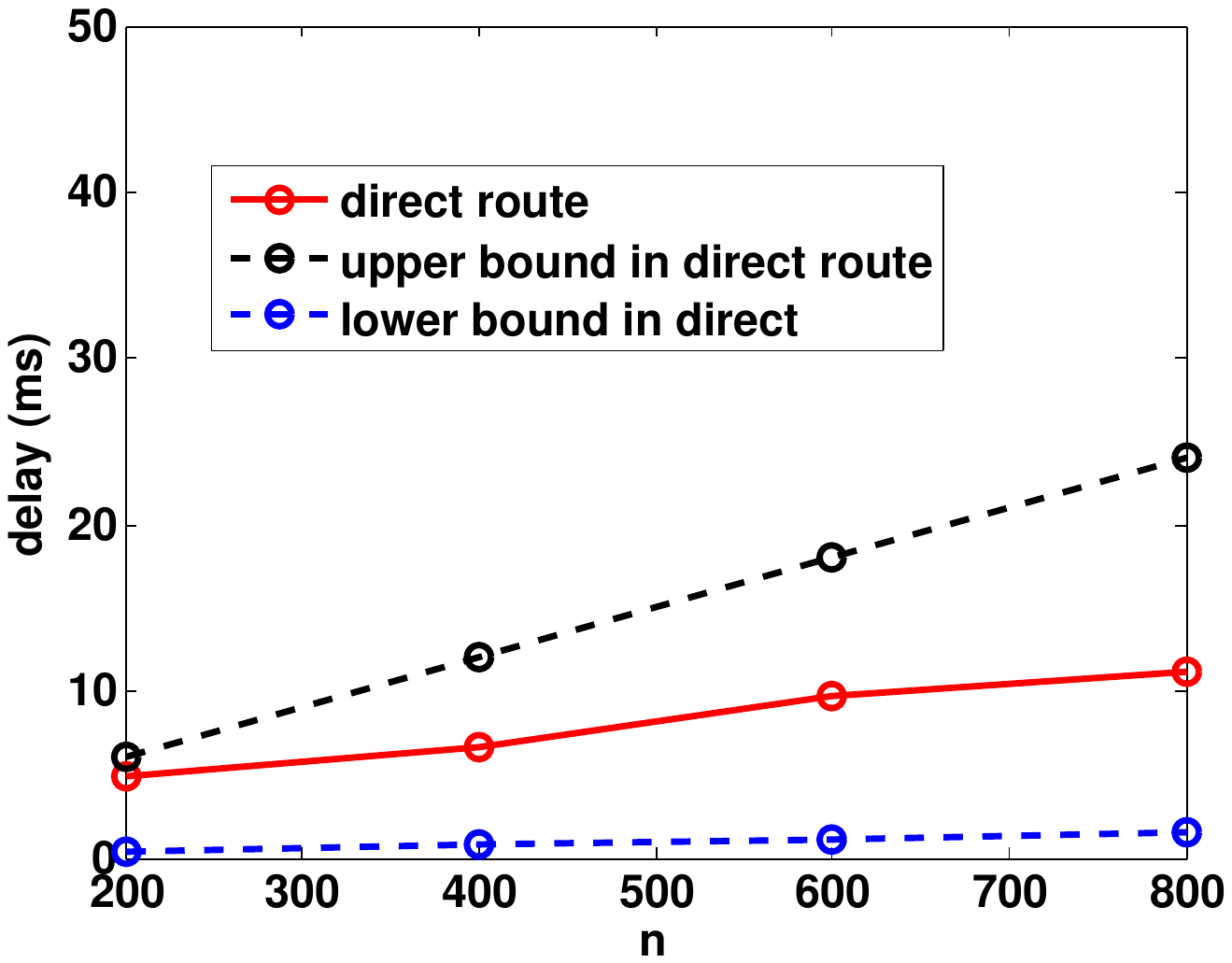}
\vspace*{-0.3cm}
\caption{ }
\label{fig:laws1}
\end{subfigure}
\begin{subfigure}[b]{0.225\textwidth}
		\includegraphics[scale=0.35, trim={4.7cm 8cm 10cm 9cm}]{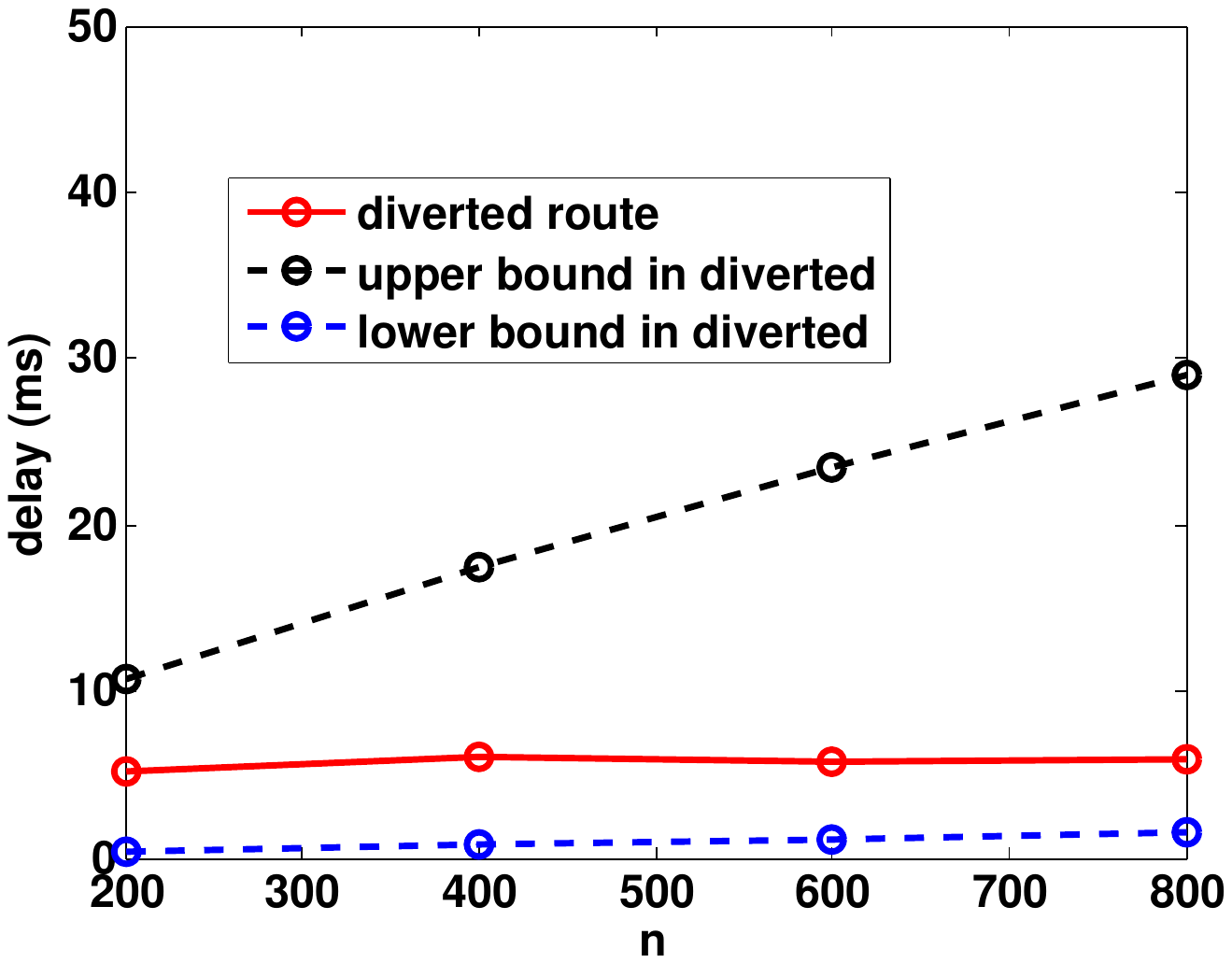}
\vspace*{-0.3cm}
\caption{ }
\label{fig:laws2}
\end{subfigure}
\begin{subfigure}[b]{0.225\textwidth}		\includegraphics[scale=0.35, trim={4.3cm 8cm 10cm 9cm}]{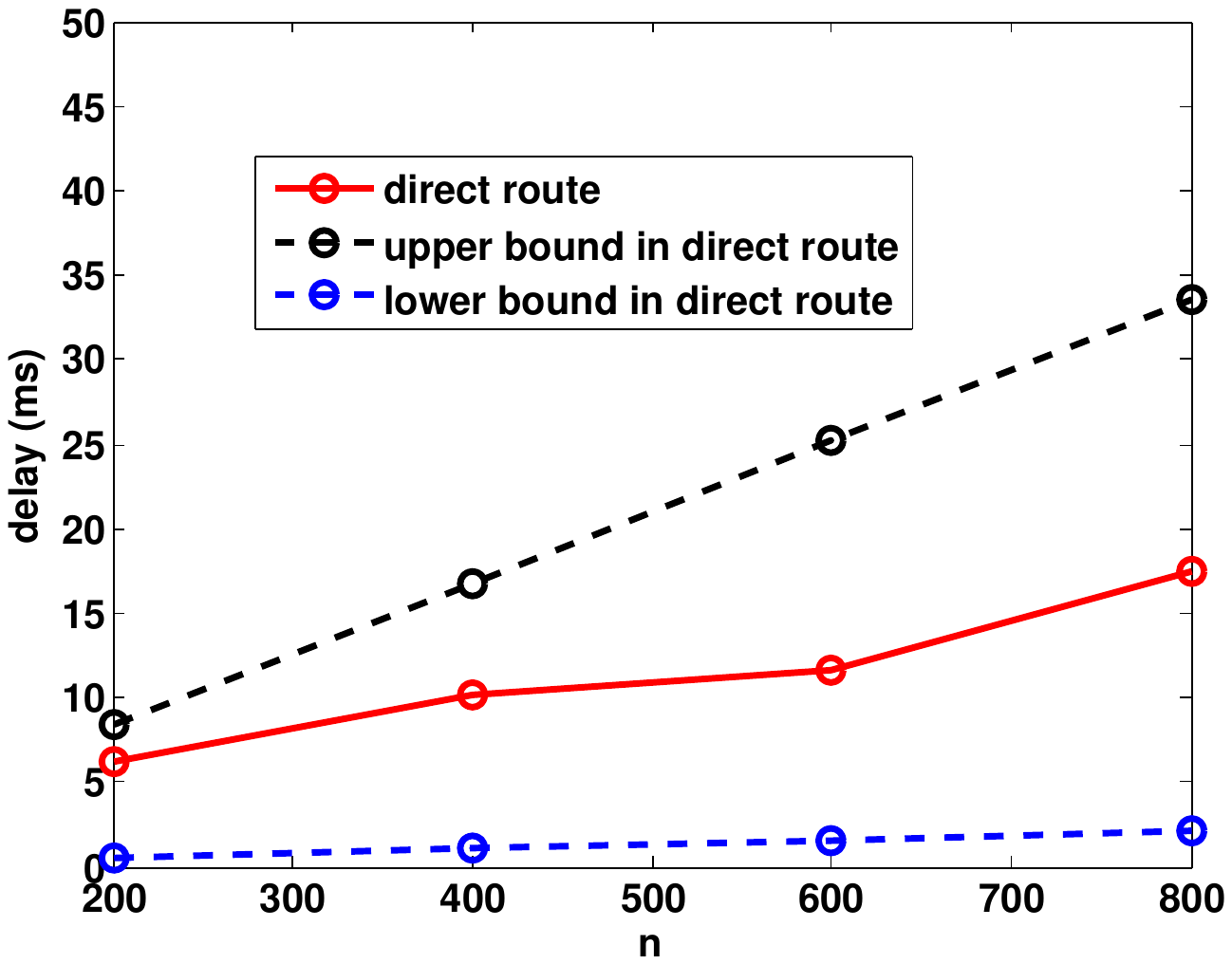}
\vspace*{-0.3cm}
\caption{ }
\label{fig:laws3}
\end{subfigure}
\begin{subfigure}[b]{0.225\textwidth}
		\includegraphics[scale=0.35, trim={3.2cm 8cm 10cm 9cm}]{df_3_75_direct_route_v_increasing.pdf}
\vspace*{-0.3cm}
\caption{ }
\label{fig:laws4}
\end{subfigure}
\caption{Validation of bounds for $v$ increasing with level: a) $d_F=3$, direct route, b) $d_F=3$, diverted route, c) $d_F=3.75$, direct route, d) $d_F=3.75$, diverted route}
\label{fig:speed_direct}
\end{figure}

The second analysis looks at the case where the speed is proportional to the inverse of the level, $v~=\frac{1}{H}$. This models the scenario where streets with a high level of importance in the city offer greater speed, like highways, but decreases with the street importance, for example, alleys. For each simulation scenario, the following vehicular mobility scenario is considered: the vehicles move with $60$ kmph on streets of level 0, $40$ kmph on level 1 and $20$ kmph on level 2, therefore the speed on the red route is $60$ kmph and on the green route $20$ kmph. Figure \ref{fig:speed_inverse} validates our bounds for this scenarios of variable speed as well. 
All the plots have been done for $\epsilon=0.1$.
\begin{figure} [ht!]
\begin{subfigure}[t]{0.225\textwidth}		\includegraphics[scale=0.35, trim={4.3cm 8cm 9cm 9cm}]{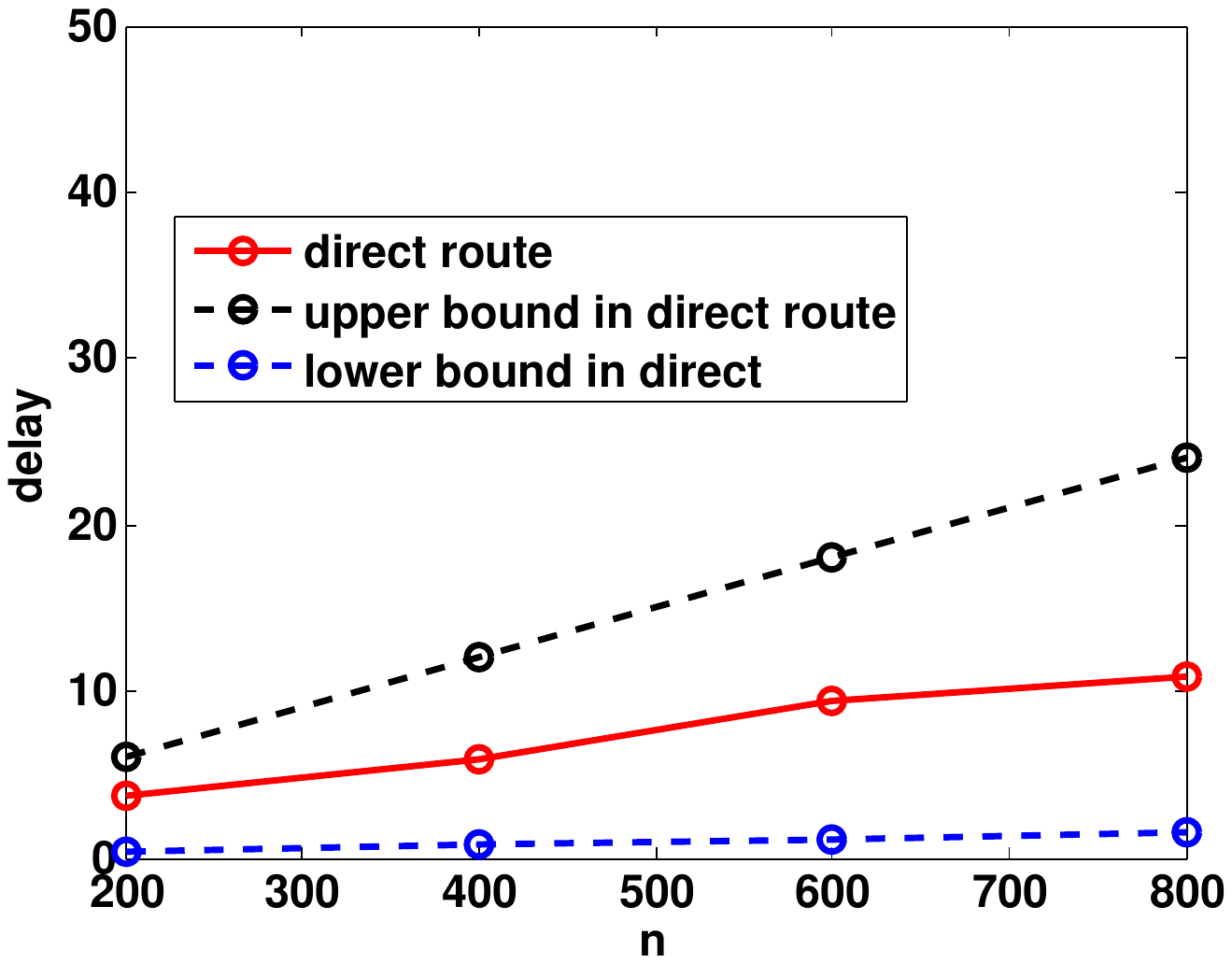}
\vspace*{-0.3cm}
\caption{ }
\label{fig:laws1}
\end{subfigure}
\begin{subfigure}[t]{0.225\textwidth}
		\includegraphics[scale=0.35, trim={4.7cm 8cm 10cm 9cm}]{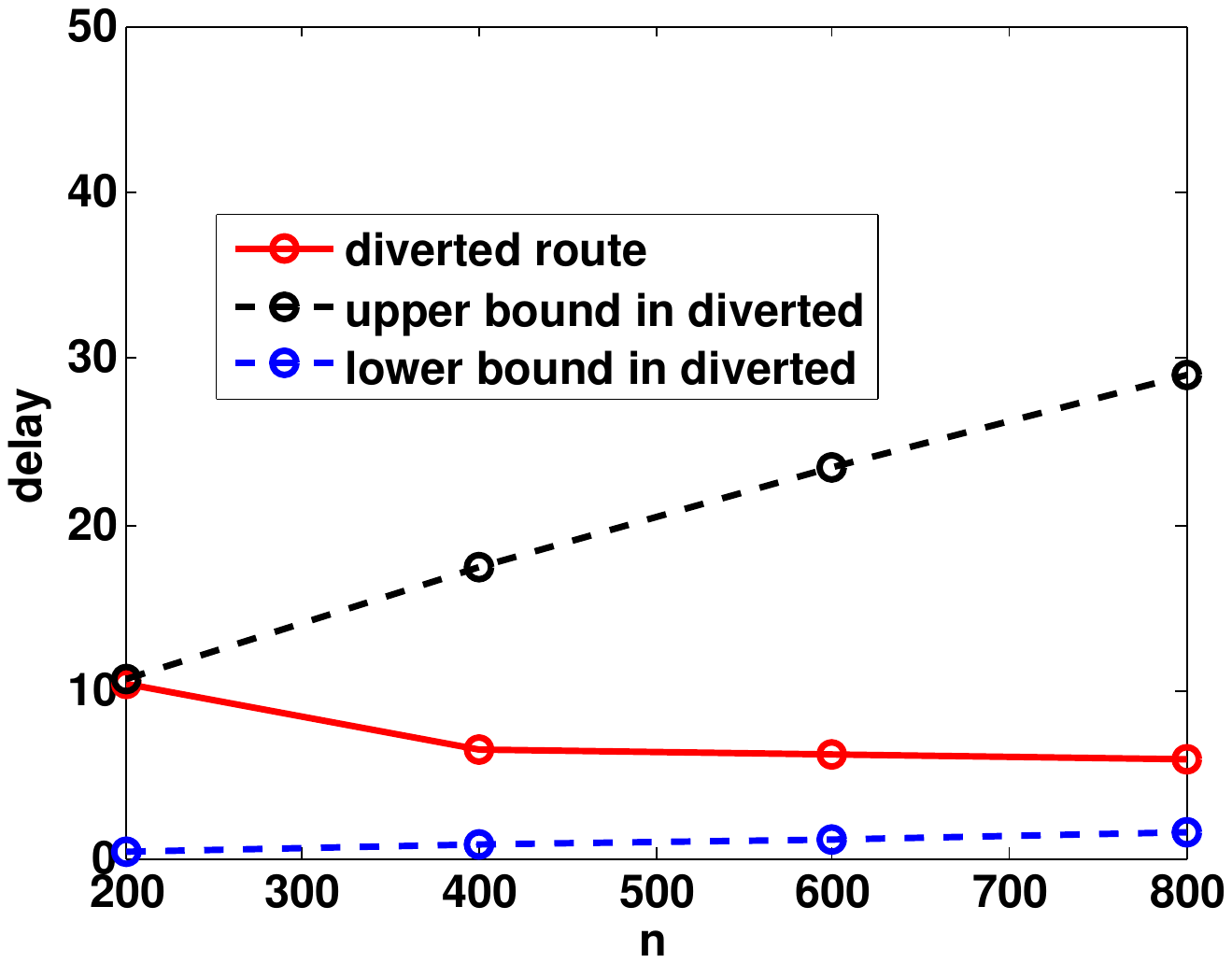}
\vspace*{-0.3cm}
\caption{ }
\label{fig:laws2}
\end{subfigure}
\begin{subfigure}[t]{0.225\textwidth}		\includegraphics[scale=0.35, trim={4.3cm 8cm 10cm 9cm}]{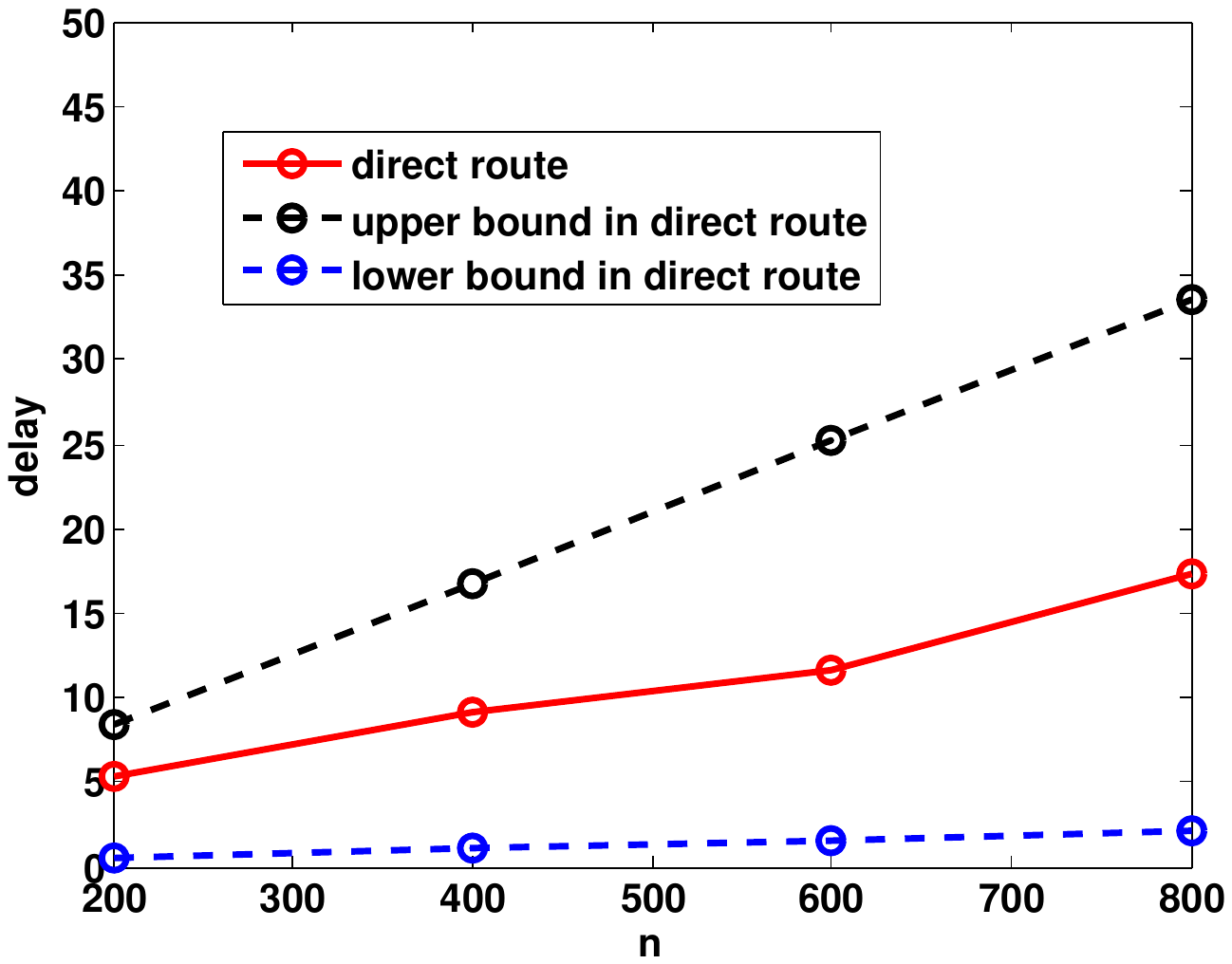}
\vspace*{-0.3cm}
\caption{ }
\label{fig:laws3}
\end{subfigure}
\begin{subfigure}[t]{0.225\textwidth}
		\includegraphics[scale=0.35, trim={3.2cm 8cm 10cm 9cm}]{df_3_75_direct_route_v_decreasing.pdf}
\vspace*{-0.3cm}
\caption{ }
\label{fig:laws4}
\end{subfigure}
\caption{Validation of bounds for $v$ decreasing with level: a) $d_F=3$, direct route, b) $d_F=3$, diverted route, c) $d_F=3.75$, direct route, d) $d_F=3.75$, diverted route}
\label{fig:speed_inverse}
\end{figure}

\begin{figure*}[h!]\label{fig:snapshots}
\begin{subfigure}[t]{0.22\textwidth}
\includegraphics[scale=0.365, trim=4cm 9cm 0cm 10cm]{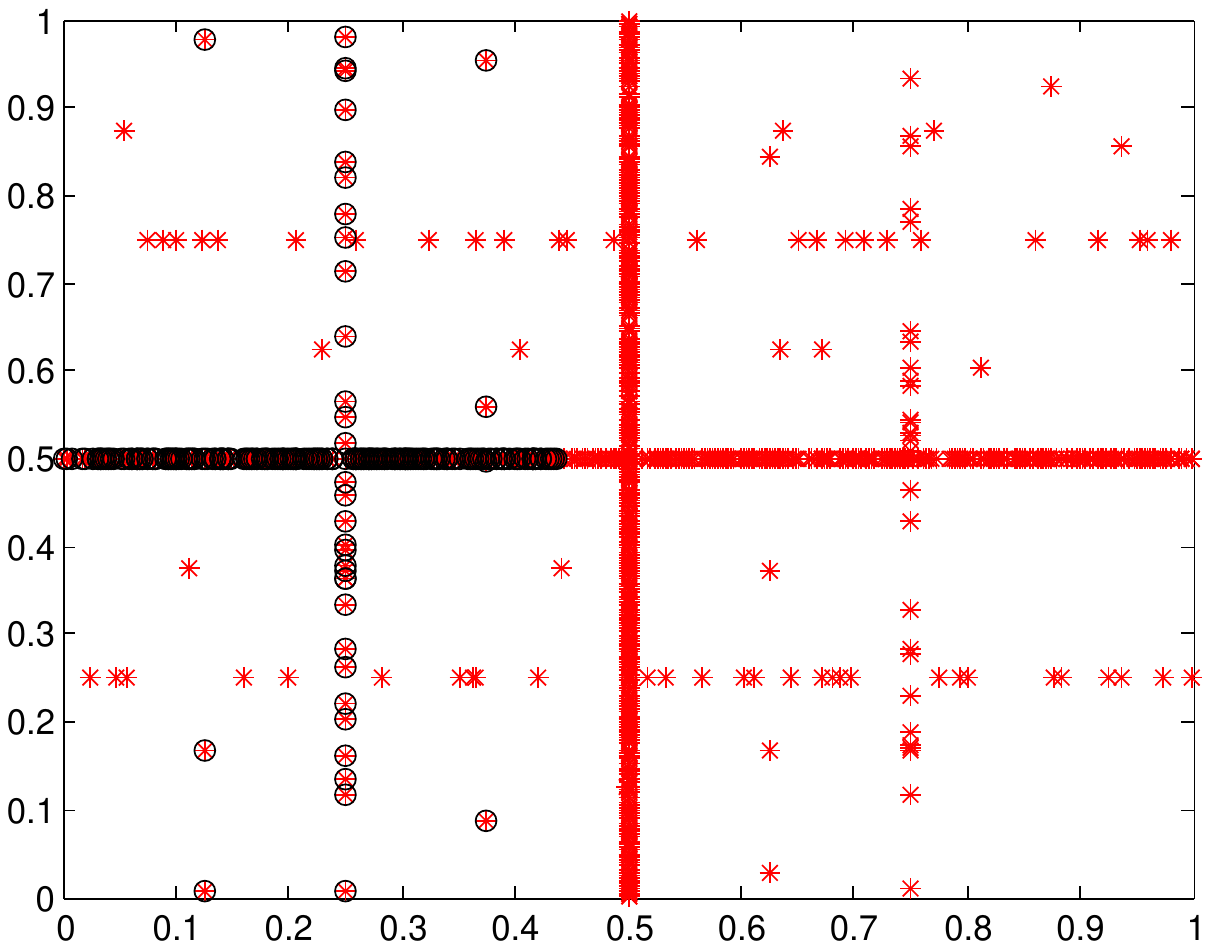}
\caption{One of the initial stages}
\vspace*{-0.12cm}
\end{subfigure}
\begin{subfigure}[t]{0.225\textwidth}
\includegraphics[scale=0.365, trim=3cm 9cm 0cm 10cm]{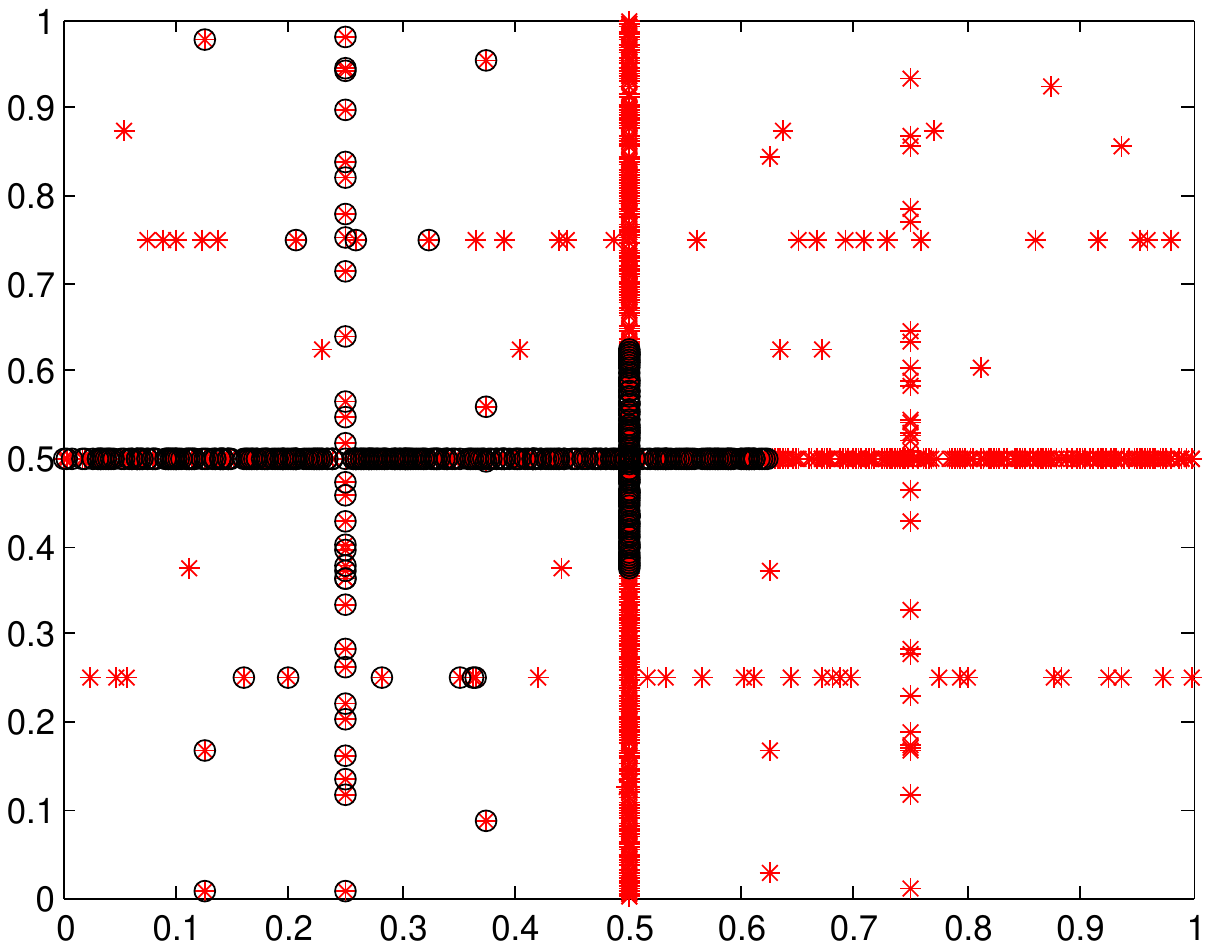}
\caption{Intermediate stage}
\end{subfigure}
\begin{subfigure}[t]{0.22\textwidth} 	
\includegraphics[scale=0.365, trim={2.3cm 9cm 0cm 10cm}]{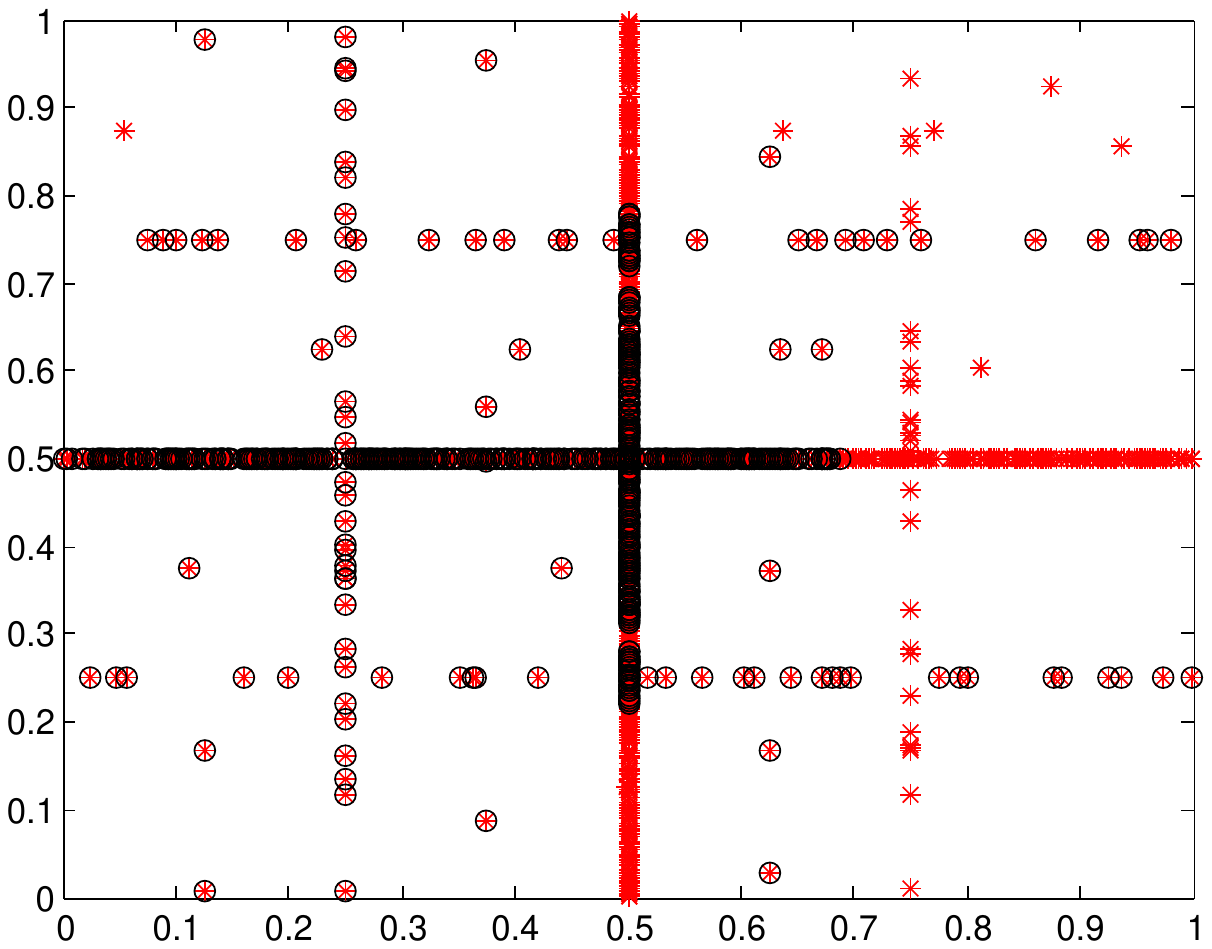}
\caption{Teleportation phenomenon
}
\vspace*{-0.3cm}
\label{fig:teleport}
\end{subfigure}
\begin{subfigure}[t]{0.22\textwidth}
\includegraphics[scale=0.365, trim=1.3cm 9cm 17cm 10cm]{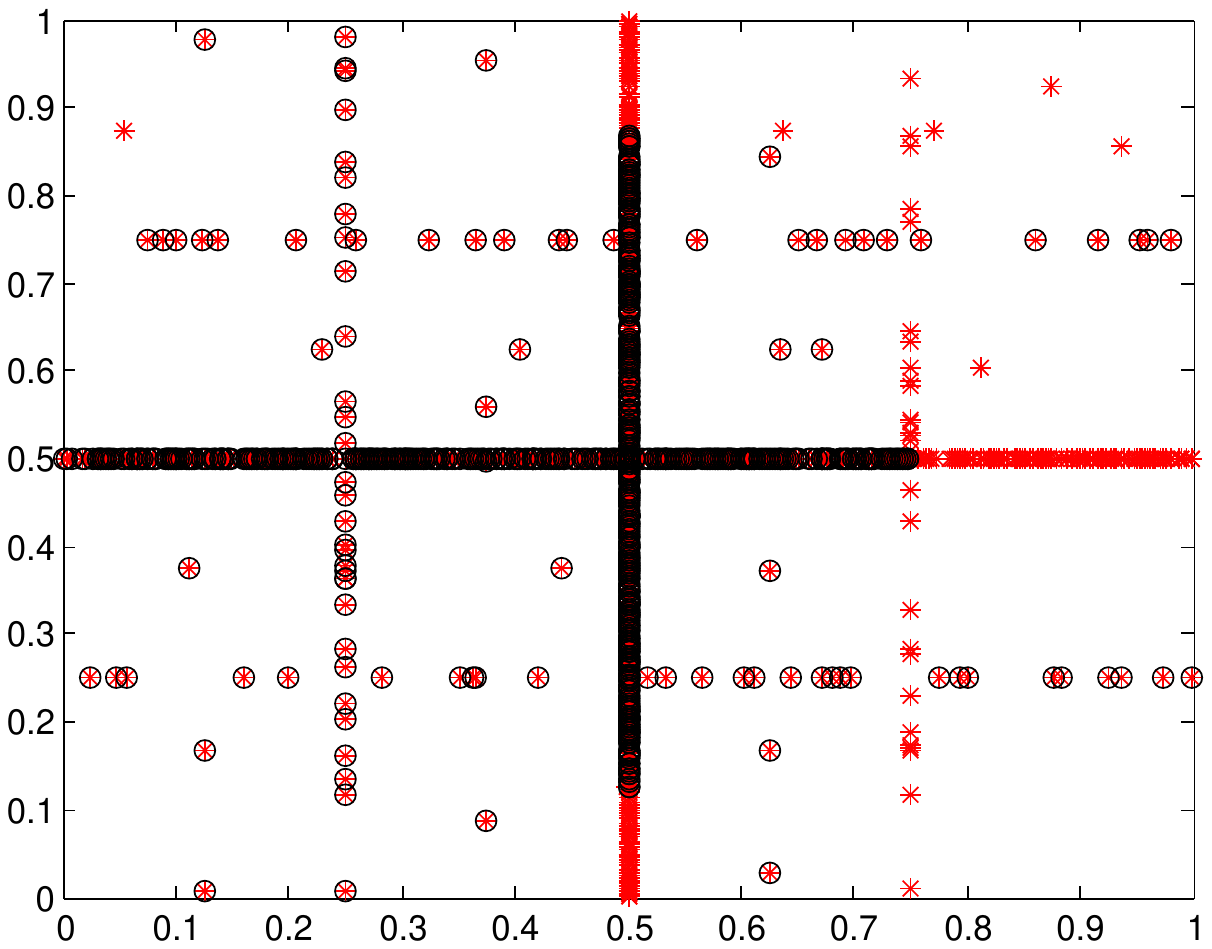}
\caption{One of the final stages}
\end{subfigure}
\caption{Snapshots of information dissemination in a Hyperfractal, healthy nodes in red `*', infected nodes in black `o'}
\label{fig:snapshots}
\end{figure*}

\subsection{Information spread under hyperfractal model and Teleportation phenomenon}

As QualNet does not allow simulating a full epidemic broadcast, we developed a discrete time event-based simulator  in Matlab which follows the model presented in Section \ref{model} in order to observe specific phenomenons that arrive when broadcasting a packet in a hyperfractal. 

For the following simulations, the levels of the hyperfractal are limited to $H=5$.  

It is well-known (see~\cite{Jacquet:DTN}) that in a two dimensional uniform Poisson point process, the information packet spreads uniformly as a full disk that grows at a constant rate, which coincides with the information propagation speed.

Interestingly, in a hyperfractal, due the canyon effect and the population distribution specific to the new model, the phenomenon is completely different. 

The simulations are performed using the following scenario: a source starts an epidemic broadcast of an information packet at time $t=0$ in a network of 1,200 nodes in a 1x1 unit square. 
The population of 1,200 nodes is distributed in the map according to a hyperfractal of dimension $d_F=5.33$. 

Figure \ref{fig:snapshots} shows different stages in the information propagation starting from a random chosen source until the complete contamination of the nodes.
The healthy nodes are depicted in red $*$ and the infected nodes are depicted in black $o$.

Although the information propagates along the streets of the network, note that it does not propagate in a uniform way (e.g., like a growing disk). The propagation follows the repartition of the population, the constraints imposed by the environment (i.e., intersections) and accelerates along the streets (i.e., canyon effect).

We can now validate the result of Corollary~\ref{coro3} on the broadcast time, as well as the lower bound of Theorem  \ref{lower} by comparing to simulations in the whole window of analysis, the complete hyperfractal map.

Figure \ref{bound_matlab} shows the simulations results obtained for a complete contamination of the network for four values of the fractal dimension: $d_F=3$, $d_F=3.3$, $d_F=3.75$, $d_F=5.75$.
\begin{figure} [h!]\label{bound_matlab}
\begin{subfigure}[b]{0.225\textwidth}		\includegraphics[scale=0.31, trim={4.3cm 9cm 9cm 10.5cm}]{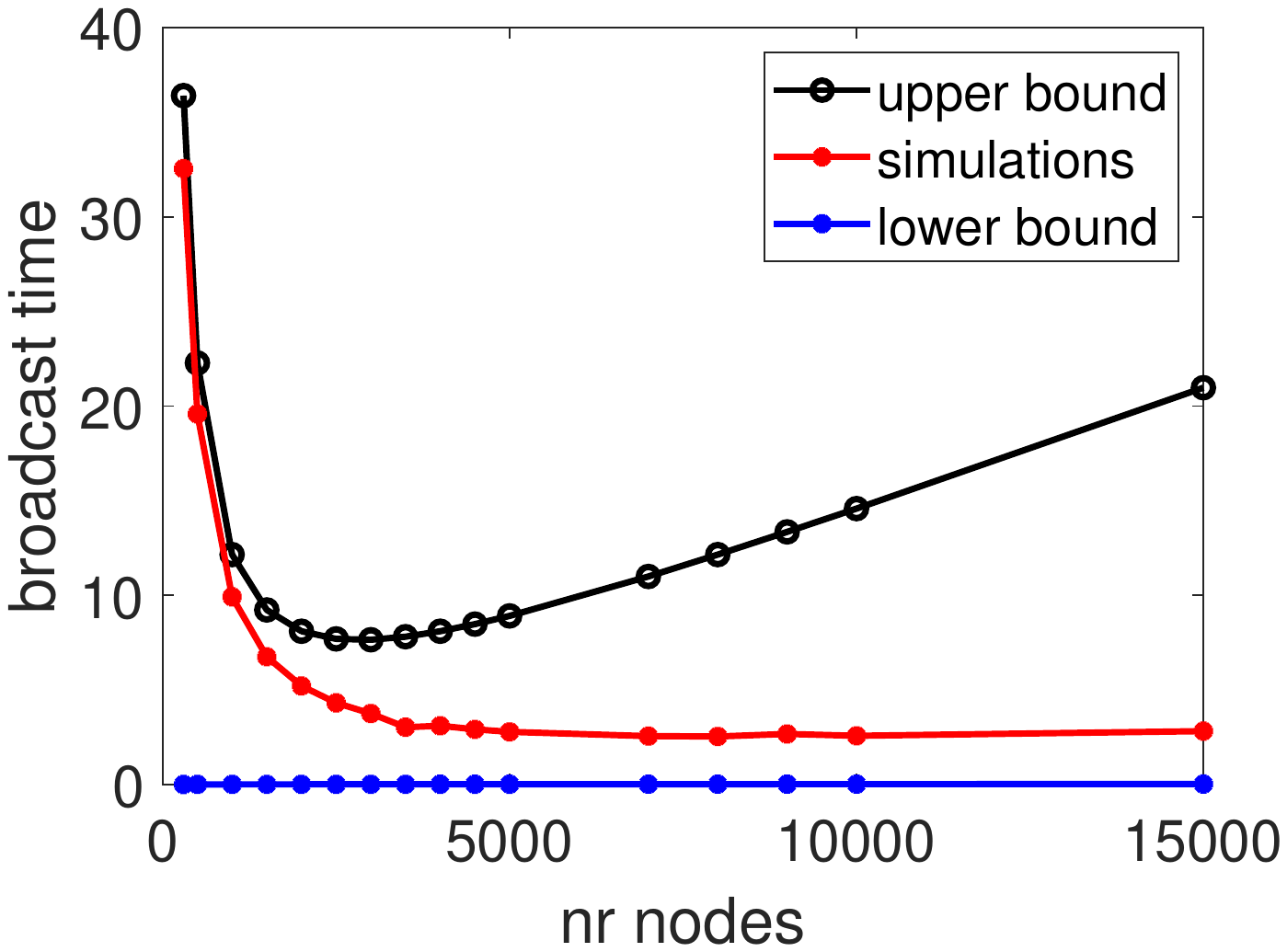}
\vspace*{-0.3cm}
\caption{ }
\label{fig:laws1}
\end{subfigure}
\begin{subfigure}[b]{0.225\textwidth}
		\includegraphics[scale=0.31, trim={4.4cm 9cm 10cm 10.5cm}]{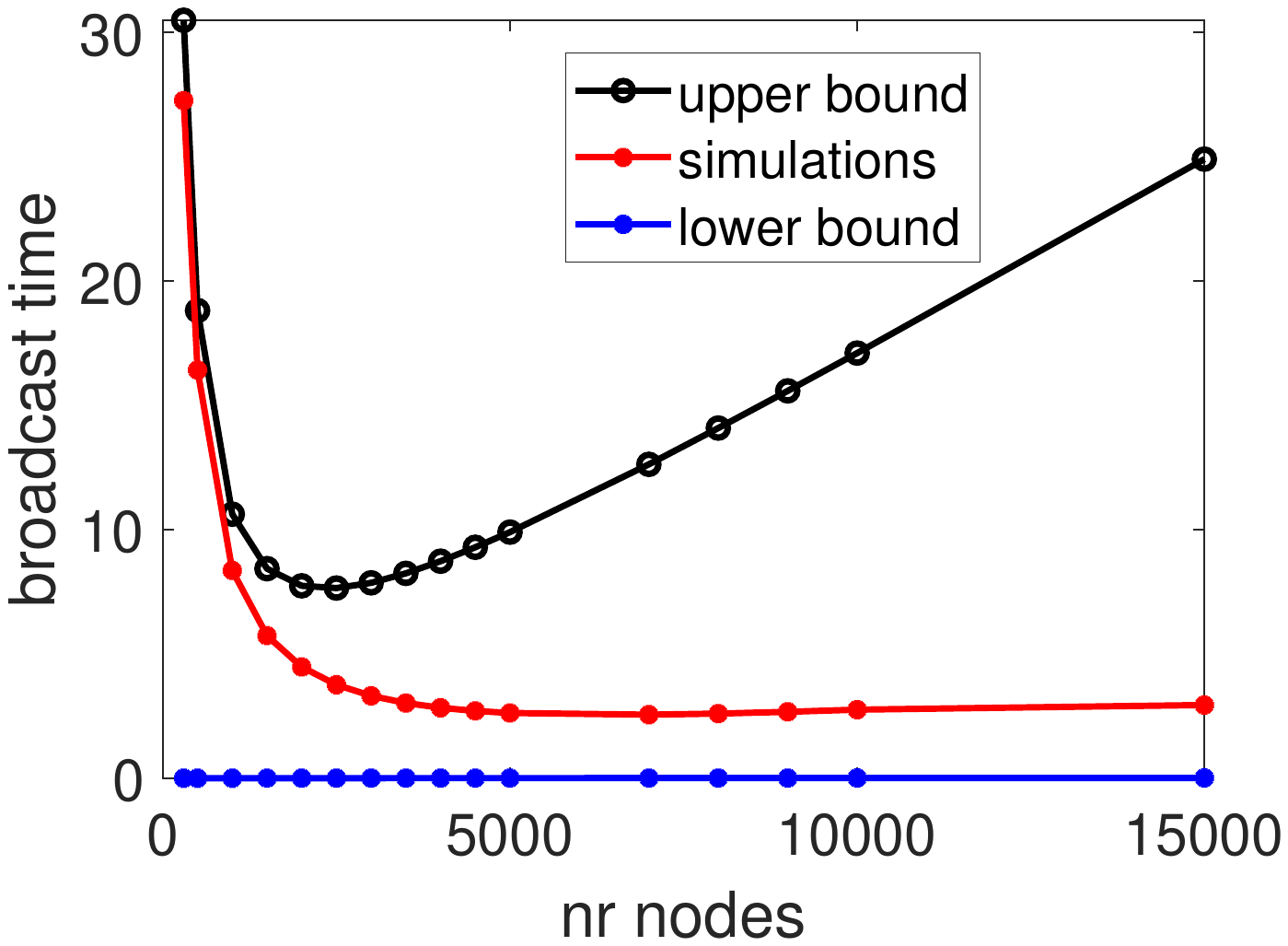}
\vspace*{-0.3cm}
\caption{ }
\label{fig:laws2}
\end{subfigure}
\begin{subfigure}[b]{0.225\textwidth}		\includegraphics[scale=0.31, trim={4cm 9cm 10cm 9cm}]{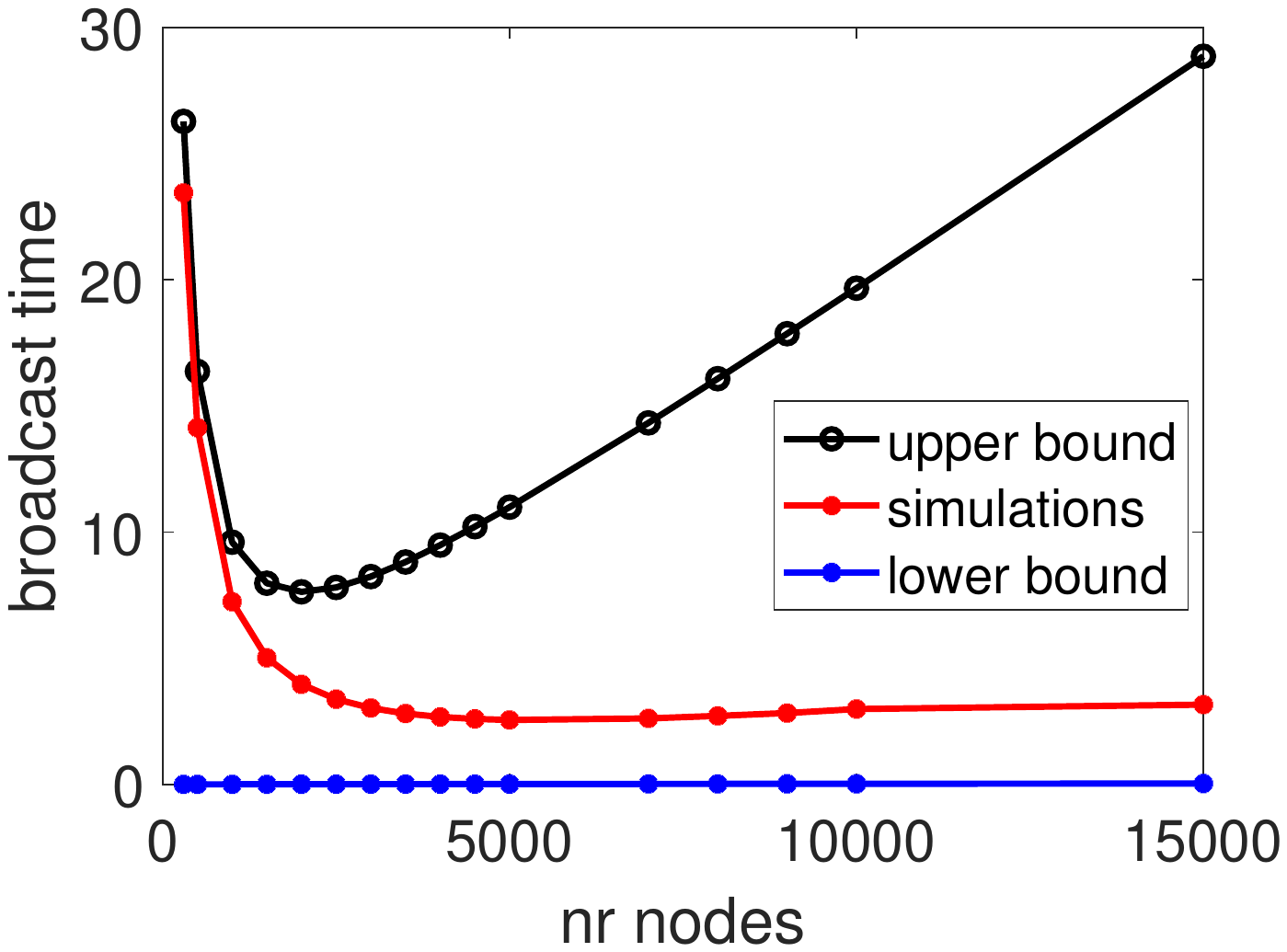}
\vspace*{-0.3cm}
\caption{ }
\label{fig:laws3}
\end{subfigure}
\begin{subfigure}[b]{0.225\textwidth}		\includegraphics[scale=0.31, trim={2.6cm 9cm 10cm 9cm}]{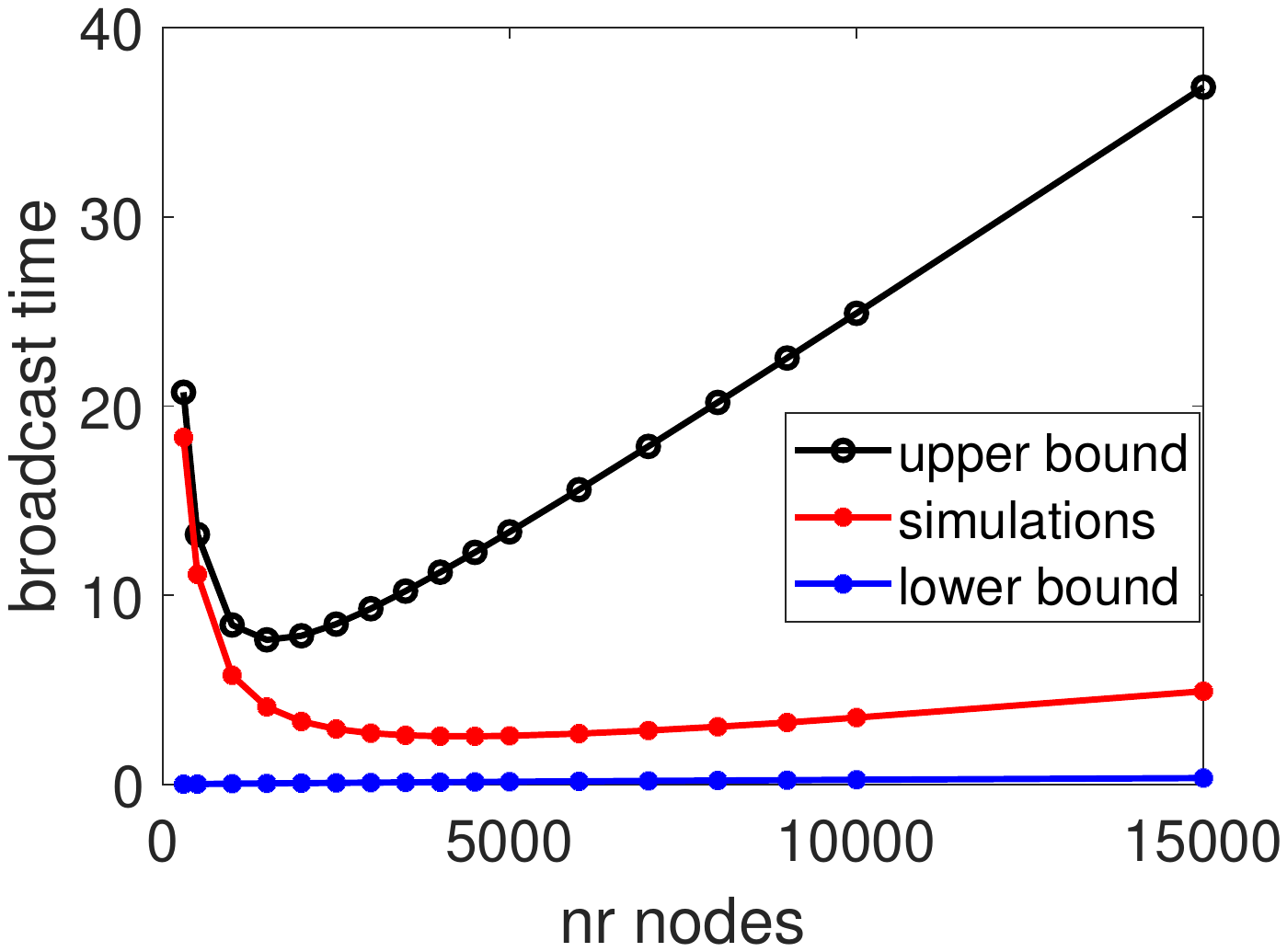}
\vspace*{-0.3cm}
\caption{ }
\label{fig:laws4}
\end{subfigure}
\caption{Broadcast time: simulation vs. theoretical results for (a) $d_F=3$, (b) $d_F=3.33$, (c) $d_F=3.75$, (d) $d_F=5.33$ }
\label{bound_matlab}
\end{figure}
Note that the lower bound follows with a good approximation the simulation results, approaching the characteristic of the simulations slope. Each of the cases shows that the broadcast time increases with the decrease of the number of points, a phenomenon which is captured successfully by the upper bound. The upper bound increases with a higher slope for big fractal dimension (see Figure \ref{fig:laws4}), converging towards the asymptotic bound of $\mathcal{O}(n^{1-\delta})$.

Let us now illustrate the propagation phenomenon that arises as a consequence of the hyperfractal distribution of nodes and the broadcast algorithm, ``the teleportation phenomenon'' introduced in  Section \ref{teleportation}. Figure \ref{fig:teleport} shows such a phenomenon in a network of $n=1,200$ nodes with $d_F=5.33$.

In Figures \ref{fig:teleport}, two contagions of infected nodes on the lines of level $0$ are highlighted. These areas are not connected to the main infected area on the line on which they originate, the line of level $H=0$. The nodes on these areas are infected by receiving the packet from nodes traveling on perpendicular lines. This gives birth to several areas of contagion. On this line, the packet is spread from all of the contamination sources that have arisen and thus the broadcast is sped up. 

This is a phenomenon that uniquely characterizes the broadcast in hyperfractal setups. 

\cleardoublepage
\section{Conclusions}\label{conclusions}

This paper provided an extended characterization of the information propagation speed of of vehicular delay tolerant mobile networks in urban scenario by using a novel model for the topology of the network and mobile vehicle locations, and by providing theoretical matching upper and lower bounds for such networks. 

These theoretical bounds are useful in order to increase our understanding of the fundamental properties and performance limits of vehicular networks in urban environments, as well as to evaluate and optimize the performance of specific routing algorithms. 

The hyperfractal model captures self-similarity as an environment characteristic and highlights interesting propagation phenomenons. The paper provided methods for the generalization of the model and data fitting and validated the theoretical results with thorough simulations in a system level commercial simulator.

\bibliographystyle{IEEEtran}
\bibliography{mybib}

\ifCLASSOPTIONcaptionsoff
\fi

\
\end{document}